\documentclass[11pt]{article}
\usepackage{graphicx,amssymb,amsmath}

\usepackage[linesnumbered, ruled]{algorithm2e}
\SetKwInput{KwResult}{Function}

\usepackage{amsfonts}
\usepackage{amsmath}
\usepackage{booktabs}
\usepackage{cite}
\usepackage{color}
\usepackage{enumerate}
\usepackage{float}
\usepackage{cite}
\usepackage{mathrsfs}
\usepackage{subfig}
\usepackage{tabularx}
\usepackage{fullpage}
\usepackage[top=0.8in, bottom=0.9in, left=0.9in, right=0.9in]{geometry}

\graphicspath{{./images/}}

\usepackage[pdftex, plainpages = false, pdfpagelabels, 
                 bookmarks=false,
                 bookmarksopen = true,
                 bookmarksnumbered = true,
                 breaklinks = true,
                 linktocpage,
                 pagebackref,
                 colorlinks = true,  
                 linkcolor = blue,
                 urlcolor  = cyan,
                 citecolor = red,
                 anchorcolor = green,
                 hyperindex = true,
                 hyperfigures
                 ]{hyperref}

\def\bbR{\mathbb{R}}

\def\calC{\mathcal{C}}
\def\calU{\mathcal{U}}
\def\calD{\mathcal{D}}

\def\calA{\Gamma}

\def\calH{\mathcal{H}}
\def\vd{\text{VD}}
\def\iq{Q_{\text{IN}}}
\def\oq{Q_{\text{OUT}}}
\def\is{S_{\text{IN}}}
\def\os{S_{\text{OUT}}}

\newtheorem{observation}{Observation}
\newtheorem{definition}{Definition}
\newtheorem{lemma}{Lemma}
\newtheorem{theorem}{Theorem}
\newtheorem{problem}{Problem}
\newtheorem{corollary}{Corollary}

\newenvironment{proof}{\noindent {\textbf{Proof:}}\rm}{\hfill $\Box$
\rm\bigskip}

\title{Dynamic Unit-Disk Range Reporting\thanks{A preliminary version of this paper will appear in {\em Proceedings of the 42nd International Symposium on Theoretical Aspects of Computer Science (STACS 2025)}. This research was supported in part by NSF under Grant CCF-2300356.}
}

\author{
Haitao Wang\thanks{Kahlert School of Computing,
University of Utah, Salt Lake City, UT 84112, USA. {\tt haitao.wang@utah.edu}}
\and
Yiming Zhao\thanks{Department of Computer Sciences,
Metropolitan State University of Denver, Denver, CO 80204, USA. {\tt yizhao@msudenver.edu}}
}

\begin{document}
\pagestyle{plain}
\date{}

\thispagestyle{empty}
\maketitle

\vspace{-0.3in}

\begin{abstract}
For a set $P$ of $n$ points in the plane and a value $r > 0$, the \emph{unit-disk range reporting} problem is to construct a data structure so that given any query disk of radius $r$, all points of $P$ in the disk can be reported efficiently. We consider the dynamic version of the problem where point insertions and deletions of $P$ are allowed. The previous best method provides a data structure of $O(n\log n)$ space that supports $O(\log^{3+\epsilon}n)$ amortized insertion time, $O(\log^{5+\epsilon}n)$ amortized deletion time, and $O(\log^2 n/\log\log n+k)$ query time, where $\epsilon$ is an arbitrarily small positive constant and $k$ is the output size. In this paper, we improve the query time to $O(\log n+k)$ while keeping other complexities the same as before. A key ingredient of our approach is a shallow cutting algorithm for circular arcs, which may be interesting in its own right. A related problem that can also be solved by our techniques is the dynamic unit-disk range emptiness queries: Given a query unit disk, we wish to determine whether the disk contains a point of $P$. The best previous work can maintain $P$ in a data structure of $O(n)$ space that supports $O(\log^2 n)$ amortized insertion time, $O(\log^4n)$ amortized deletion time, and $O(\log^2 n)$ query time. Our new data structure also uses $O(n)$ space but can support each update in $O(\log^{1+\epsilon} n)$ amortized time and support each query in $O(\log n)$ time.
\end{abstract}


\section{Introduction}
\label{sec:introduction}
Range searching is a fundamental problem and has been studied extensively in computational geometry~\cite{ref:AgarwalRa17,ref:AgarwalSi17,ref:MatousekGe94}. In this paper, we consider a dynamic range reporting problem regarding disks of fixed radius, called {\em unit disks}. 

Given a set $P$ of $n$ points in the plane, the {\em unit-disk range reporting} problem (or UDRR for short) is to construct a data structure to report all points of $P$ in any query unit disk. 
The problem is also known as the {\em fixed-radius neighbor problem} in the literature~\cite{ref:BentleyA79,ref:ChazelleAn83,ref:ChazelleNe86,ref:ChazelleOp85}.
Chazelle and Edelsbrunner~\cite{ref:ChazelleOp85} constructed a data structure of $O(n)$ space that can answer each query in $O(\log n+k)$ time, where $k$ is the output size; their data structure can be constructed in $O(n^2)$ time.
By a standard lifting transformation~\cite{ref:deBergCo08}, the problem can be reduced to the half-space range reporting queries in 3D; this reduction also works if the radius of the query disk is arbitrary.
Using Afshani and Chan's 3D half-space range reporting data structure~\cite{ref:AfshaniOp09}, one can construct a data structure of $O(n)$ space with $O(\log n + k)$ query time, while the preprocessing takes $O(n \log n)$ expected time since it invokes Ramos' algorithm~\cite{ref:RamosOn99} to construct shallow cuttings for a set of planes in 3D. 
Chan and Tsakalidis~\cite{ref:ChanOp16} later presented an $O(n\log n)$-time deterministic shallow cutting algorithm. Combining the framework in~\cite{ref:AfshaniOp09} with the shallow cutting algorithm~\cite{ref:ChanOp16}, one can build a data structure of $O(n)$ space in $O(n \log n)$ deterministic time that can answer each UDRR query in $O(\log n + k)$ time. 


We consider the dynamic UDRR problem in which point insertions and deletions of $P$ are allowed. By the lifting transformation, the problem can be reduced to dynamic halfspace range reporting in 3D~\cite{ref:ChanDy20,ref:ChanRa00,ref:ChanTh12}, which also works for query disks of arbitrary radii. Using the currently best result of dynamic halfspace range reporting~\cite{ref:BergDy23}, one can obtain a data structure of $O(n\log n)$ space that supports $O(\log^{3+\epsilon}n)$ amortized insertion time, $O(\log^{5+\epsilon}n)$ amortized deletion time, and $O(\log^2 n/\log\log n+k)$ query time, where $\epsilon$ is an arbitrarily small positive constant and $k$ is the output size. 

\paragraph{Our result.}
In this paper, we achieve the optimal $O(\log n+k)$ query time, while the space of the data structure and the update time complexities are the same as above. 

A byproduct of our techniques is a static data structure of $O(n)$ space that can be built in $O(n\log n)$ time and support $O(\log n+k)$ query time. This matches the above result of~\cite{ref:AfshaniOp09,ref:ChanOp16}. But our method is much simpler. Indeed, the algorithm of~\cite{ref:AfshaniOp09,ref:ChanOp16} involves relatively advanced geometric techniques like computing shallow cuttings for the planes in 3D, planar graph separators, etc. Our algorithm, on the contrary, relies only on elementary techniques (the most complicated one might be a fractional cascading data structure~\cite{ref:ChazelleFr86,ref:ChazelleFr862}). One may consider our algorithm a generalization of the classical 2D half-plane range reporting algorithm of Chazelle, Guibas, and Lee~\cite{ref:ChazelleTh85}. 

Our techniques may also be useful for solving other problems related to unit disks. In particular, we can obtain an efficient algorithm for the dynamic {\em unit-disk range emptiness queries}. For a dynamic set $P$ of points in the plane, we wish to determine whether a query unit disk contains any point of $P$ (and if so, return such a point as an ``evidence''). The previous best solution is to use a dynamic nearest neighbor search data structure~\cite{ref:ChanDy20}. Specifically, we can have a data structure of $O(n)$ space that supports $O(\log^2 n)$ amortized insertion time, $O(\log^4n)$ amortized deletion time, and $O(\log^2 n)$ query time. Using our techniques, we obtain an improved data structure of $O(n)$ space that supports both insertions and deletions in $O(\log^{1+\epsilon} n)$ amortized time and supports queries in $O(\log n)$ time.


\paragraph{Our approach.}
We first discuss our static data structure. 
We use a set of $O(n)$ grid cells (each of which is an axis-parallel rectangle) to capture the proximity information for the points of $P$, such that the distance between any two points in the same cell is at most 1. For a query unit disk $D_q$
whose center is $q$, points of $P$ in the cell $C$ that contains
$q$ can be reported immediately. The critical part is handling other cells that contain points of $P\cap D_q$. The number of such cells is
constant and each of them is separated from $C$ (and thus from $q$) by an
axis-parallel line. The problem thus boils down to the following subproblem: Given a set $Q$ of points
in a grid cell $C'$ above a horizontal line $\ell$, report the points of $Q$ in any query unit disk whose
center is below $\ell$. A point $p\in Q$ is in $D_q$ if and only if $q$ lies in the unit disk $D_p$
centered at $p$, or equivalently, $q$ is above the arc of the boundary of $D_p$
below $\ell$. Let $\calA$ denote the set of all such arcs for all points $p\in Q$. To find
the points of $Q$ in $D_q$, it suffices to report the arcs of $\calA$ below $q$.
To tackle this problem, we follow the same framework as that for the 2D
half-plane range reporting algorithm~\cite{ref:ChazelleTh85}, by constructing lower envelope layers of $\calA$ and building a fractional cascading data structure on them~\cite{ref:ChazelleFr86,ref:ChazelleFr862}. 


To make the data structure dynamic, we first derive a data structure to maintain the grid cells dynamically so that each update (point insertions/deletions) can be handled in $O(\log n)$ amortized time. This dynamic data structure could be of independent interest. Next, we develop a data structure to dynamically maintain arcs of $\calA$ to support the arc-reporting queries (i.e., given a query point, report all arcs of $\calA$ below it). To this end, we cannot use the fractional cascading data structure anymore because it is not amenable to dynamic changes. Instead, we adapt the techniques for dynamically maintaining a set of lines to answer line-reporting queries (i.e., given a query point, report all lines below the query point; its dual problem is the halfplane range reporting queries)~\cite{ref:BergDy23,ref:ChanA10,ref:ChanTh12,ref:ChanDy20,ref:KaplanDy20}. To make these techniques work for the arcs of $\calA$ in our problem, we need an efficient shallow cutting algorithm for $\calA$. For this, we adapt the algorithm in \cite{ref:ChanOp16} for lines and derive an $O(|\calA|\log |\calA|)$ time shallow cutting algorithm for $\calA$. As shallow cuttings have many applications, our algorithm may be interesting in its own right.

\paragraph{Outline.}
The rest of the paper is organized as follows. In Section~\ref{sec:pre}, we introduce notation and a {\em conforming coverage set} of grid cells to capture the proximity information for points of $P$. In particular, in Section~\ref{sec:pre} we state a lemma on dynamically maintaining the grid cells. The proof of the lemma, which is quite lengthy and technical, is presented in
Section~\ref{sec:lemdynamiccellproof}. Section~\ref{sec:dynamicreport} discusses our dynamic data structure for the UDRR problem. A main subproblem of it is solved in Section~\ref{sec:lemdynamiclineUDRR}. A key ingredient of our method is an efficient algorithm for computing shallow cuttings for arcs; this algorithm is presented in Section~\ref{sec:algoshallow}. 
Our static UDRR data structure is described in Section~\ref{sec:udrr}. The algorithm uses a subroutine that computes layers of lower envelopes of circular arcs; the subroutine is presented in Section~\ref{sec:ComputingLayers}. 
Section~\ref{sec:ConcludingRemarks} concludes the paper and demonstrates that our techniques may be used to solve other related problems, such as dynamic unit-disk range emptiness queries.

\section{Preliminaries}
\label{sec:pre}

We define some notation that will be used throughout the paper. 

A {\em unit disk} refers to a disk of radius $1$. A {\em unit circle} is defined similarly.  Unless otherwise stated, a circular arc or an arc refers to a circular arc of radius $1$. For a circular arc $\gamma$, we call the circle that contains $\gamma$ the {\em underlying circle} of $\gamma$ and call the disk whose boundary contains $\gamma$ the {\em underlying disk}.

For any point $q$, let $D_q$ denote the unit disk centered at $q$. For any region $R$ and any set $P$ of points in the plane, let $P(R)$ denote the subset of points of $P$ inside $R$, i.e., $P(R)=P\cap R$. For any region $R$ in the plane, we use $\partial R$ to denote its boundary, e.g., if $R$ is a disk, then $\partial R$ is its bounding circle. 

Unless otherwise stated, $\epsilon$ refers to an arbitrarily small positive constant. Depending on the context, we often use $k$ to denote the output size of a reporting query. 
For any point $p$ in the plane, we use $x(p)$ and $y(p)$ to denote the $x$ and $y$-coordinates of $p$, respectively.

\subsection{Conforming coverage of $\boldsymbol{P}$}

Let $P$ be a set of $n$ points in the plane. We wish to have a data structure for $P$ to answer the following {\em unit-disk range reporting queries}: Given a query unit disk $D$, report $P(D)$, i.e., the points of $P$ in $D$.

As discussed in Section~\ref{sec:introduction}, our method (for both static and dynamic problems) relies on a set of grid cells to capture the proximity information for the points of $P$. The technique of using grids has been widely used in various algorithms for solving
problems in unit-disk graphs~\cite{ref:WangRe23,ref:WangNe20, ref:ChanAl16, ref:WangAn23, ref:WangCo22,ref:WangUn23}. However, the difference here is that we need to handle updates to $P$ and therefore our grid cells will be dynamically changed as well. To resolve the issue, our definition of grid cells is slightly different from the previous work. Specifically, we define a {\em conforming coverage} set of cells for $P$ in the following. 
\begin{definition}\label{def:cell}(Conforming Coverage)
A set $\calC$ of cells in the plane is called a {\em conforming coverage} for $P$ if the following conditions hold. 
\begin{enumerate}
    \item Each cell of $\calC$ is an axis-parallel rectangle of side lengths at most $1/2$. This implies that the distance of every two points in each cell is at most $1$.
    \item The union of all cells of $\calC$ covers all the points of $P$. 
    \item Every two cells are separated by an axis-parallel line. 
    \item Each cell $C\in \calC$ is associated with a subset $N(C)\subseteq \calC$ of $O(1)$ cells (called {\em neighboring cells} of $C$) such that for any point $q\in C$, $P(D_q) \subseteq \bigcup_{C'\in N(C)}P(C')$. 
    \item For any point $q$, if $q$ is not in any cell of $\calC$, then $P \cap D_q = \emptyset$.
\end{enumerate}  
\end{definition}

To solve the static problem, after computing a conforming coverage set of cells for $P$, we never need to change it in the future. As such, the following lemma from \cite{ref:WangUn23} suffices. 
\begin{lemma}
\label{lem:grid}{\em \cite{ref:WangUn23}}
 \begin{enumerate}
     \item
   A conforming coverage set $\calC$ of size $O(n)$, along with $P(C)$ and $N(C)$ for all cells $C \in \calC$, can be computed in $O(n\log n)$ time and $O(n)$ space.
 \item
    With $O(n \log n)$ time and $O(n)$ space preprocessing, given any point $q$, we can do the following in $O(\log n)$ time: Determine whether $q$ is in a cell $C$ of $\calC$, and if so, return $C$ and $N(C)$.
 \end{enumerate}
\end{lemma}

However, for the dynamic problem, due to the updates of $P$, the conforming coverage set also needs to be maintained dynamically. For this, we have the following lemma. 
\begin{lemma}\label{lem:dynamiccell}
A conforming coverage set $\calC$ of $O(n)$ cells for $P$ can be maintained in $O(n)$ space (where $n$ is the size of the current set $P$) such that each point insertion of $P$ can be handled in $O(\log n)$ worst-case time, each point deletion can be handled in $O(\log n)$ amortized time, and the following query can be answered in $O(\log n)$ time: Given a query point $q$, determine whether $q$ is in a cell $C$ of $\calC$, and if so, return $C$ and $N(C)$.
\end{lemma}

Since the proof of Lemma~\ref{lem:dynamiccell} is lengthy and technical, we devote Section~\ref{sec:lemdynamiccellproof} to it. Roughly speaking, if a point $p$ is inserted to $P$, then at most $O(1)$ cells will be added to $\calC$ and $p$ will eventually be inserted into $P(C)$ for the cell $C\in\calC$ containing $p$. If a point $p$ is deleted from $P$, the deletion boils down to the deletion of $p$ from $P(C)$ for the cell $C\in\calC$ containing $p$. We do not remove cells from $\calC$. Instead, we reconstruct the entire data structure after $n/2$ deletions; this guarantees that the size of $\calC$ is always $O(n)$. See Section~\ref{sec:lemdynamiccellproof} for the details. 

\section{Proving Lemma~\ref{lem:dynamiccell}: Maintaining a conforming coverage set dynamically}
\label{sec:lemdynamiccellproof}
In this section, we prove Lemma~\ref{lem:dynamiccell}. As stated in Lemma~\ref{lem:grid}, an algorithm is already given in \cite{ref:WangUn23} to compute a conforming coverage set $\calC$ in $O(n\log n)$ time and $O(n)$ space. However, the set $\calC$ computed by that algorithm is not quite suitable for the dynamic setting. Instead, based on that algorithm, we propose a new algorithm that computes a new conforming coverage set $\calC$ of size $O(n)$ along with a data structure that is amenable to point updates. In fact, deletions are relatively easy to deal with as discussed above. The challenge is to handle insertions. In what follows, we first present our algorithm for computing $\calC$ and our data structure to maintain it for a given set $P$ of $n$ points. Then we will discuss insertions and deletions. 

\subsection{Computing $\boldsymbol{\calC}$ and the data structure}


For any vertical line $\ell$, we use $x(\ell)$ to denote its $x$-coordinate. 

We sort all the points of $P$ as $p_1,p_2,\ldots,p_n$, in ascending order by their $x$-coordinates. 

\paragraph{Computing point-zones.}
We compute a set of $O(n)$ disjoint vertical strips in the plane, called {\em point-zones}, each bounded by two vertical lines and containing at least one point of $P$. 

Starting from $p_1$, we sweep the plane by a vertical line $\ell$. 
We maintain an invariant that $\ell$ is in the current vertical point-zone whose left bounding line, denoted by $\ell_1$, has already been computed and whose right bounding line $\ell_2$ is to be determined. Initially, we place a vertical line at $x(p_1)-7/4$ as the left bounding line of the first point-zone. Suppose that $\ell$ is at a point $p_i$. 
If $i<n$ and $x(p_{i+1})-x(p_i)\leq 5$, then we move $\ell$ to $p_{i+1}$. 
Otherwise, we place the right bounding line $\ell_2$ at $x(p_i)+2+x'$, where $x'$ is the
smallest nonnegative value such that $x(p_i)+2+x'-x(\ell_1)$ is a multiple
of $1/2$; this produces a new point-zone. Observe that the width of the point-zone (i.e., $|x(\ell_2)-x(\ell_1)|$) is $O(m)$ if $m$ is the number of points of $P$ in the point-zone. As such, the above value $x'$ can be easily computed in $O(m)$ time. 
Next, if $i=n$, the algorithm stops; otherwise, we move $\ell$ to $p_{i+1}$ to compute the next point-zone following the same algorithm (e.g., we start placing the left bounding line at $x(p_{i+1})-7/4$). 

The above algorithm, which runs in $O(n)$ time, computes at most $n$ vertical point-zones that are
pairwise-disjoint. We call the region between two adjacent point-zones {\em a vertical gap-zone}. The region to the left of the leftmost point-zone and the region to the right of the rightmost point-zone are also vertical gap-zones. 
According to the algorithm, the following properties hold: (1) For any point $q$ is in a gap-zone, $P\cap D_q=\emptyset$; (2) if a vertical point-zone contains $m$ points of $P$, then the width of the point-zone is $O(m)$, implying that the sum of the widths of all vertical point-zones is $O(n)$; (3) the width of each vertical point-zone is a multiple of $1/2$; (4) for any point $p\in P$, the distance between $p$ and each bounding line of $Z$ is at least $7/4$, where $Z$ is the point-zone containing $p$. 

Similarly, by sweeping a line from top to bottom, we compute a set of $O(n)$ {\em horizontal point-zones} and {\em horizontal gap-zones}; see Fig.~\ref{fig:griddynamics}. The following properties hold: (1) If a point $q$ is in a horizontal gap-zone, then $P\cap D_q=\emptyset$; (2) if a horizontal point-zone contains $m$ points of $P$, then the height of the point-zone is $O(m)$, implying that the sum of the heights of all horizontal point-zones is $O(n)$; (3) the height of each horizontal point-zone is a multiple of $1/2$; (4) for any point $p\in P$, the distance between $p$ and each bounding line of $Z$ is at least $7/4$, where $Z$ is the horizontal point-zone containing $p$. 

    \begin{figure}[t]
        \centering
        \includegraphics[height=2.0in]{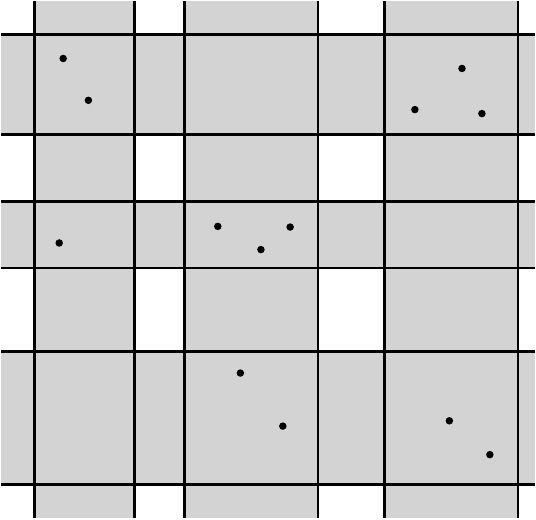}
        \caption{The point-zones lie in the grey area. The black dots are points of $P$.}
        \label{fig:griddynamics}
        \vspace{-0.15in}
    \end{figure}

\paragraph{Forming a grid $\boldsymbol{G}$.}
For each vertical point-zone, since its width is a multiple of $1/2$, we further add vertical lines to partition it into vertical regions of width $1/2$ each; due to the above property (4) of the vertical point-zones, we need add $O(m)$ vertical lines if the point-zone contains $O(m)$ points of $P$. Let $L_v$ denote the set of all these vertical lines. We also add the vertical point-zone bounding lines to $L_v$. Clearly, we have $|L_v|=O(n)$. For differentiation, we use {\em non-boundary vertical partition lines} to refer to the lines of $L_v$ that are not point-zone bounding lines. 

Similarly, we partition each horizontal point-zone into regions of heights $1/2$ by adding {\em non-boundary} horizontal partition lines. Define $L_h$ to be the set of all these lines, including all horizontal point-zone bounding lines. 


The lines of $L_v\cup L_h$ together partition the plane into a grid $G$ of $O(n^2)$ cells, each of which is an (axis-parallel) rectangle. We say that a cell of $G$ is a {\em regular cell} if it is contained in both a vertical point-zone and a horizontal point-zone, and it is a {\em gap cell} otherwise. According to the above discussion, we have the following properties about $G$: (1) If a point $q$ is in a gap cell, then $P\cap D_q=\emptyset$; (2) the length of each edge of a regular cell is at most $1/2$ (note that it is exactly equal to $1/2$ according to our above algorithm, but after we handle point insertions, it is possible that an edge length of a regular cell is smaller than $1/2$). 

All cells of $G$ between two adjacent vertical lines of $L_v$ form a {\em column} of $G$. A column is a {\em regular column} if it is in a vertical point-zone, and it is a {\em gap column} otherwise. All cells of $G$ between two adjacent horizontal lines of $L_h$ form a {\em row} of $G$. Similarly, we have regular rows and gap rows. The width of each regular column is at most $1/2$; so is the height of each regular row. 

For each cell $C$ of $G$, we use $\boxplus_C$ to denote the area of $7\times 7$ cells of $G$ with $C$ at the center; note that $\boxplus_C$ may contain less than 49 cells if it is close to the boundary of $G$. In fact, whenever we use $\boxplus_C$ in the the algorithm, $C$ is always a cell that contains a point of $P$. For such ``non-empty'' cells (i.e., cells that contain points of $P$), our algorithm always guarantees that $\boxplus_C$ contains exactly 49 cells. Indeed, this is true in our above algorithm for computing $L_v$ and $L_h$ due to the above properties (4) of the point-zones. This will also hold when the grid $G$ is changed due to point insertions.

Our algorithm ensures the following key property {\em} for $\boxplus_C$ (even after $G$ changes due to point updates): For any point $q\not\in \boxplus_C$, $D_q\cap C=\emptyset$. This is the main reason we introduce the notation $\boxplus_C$. 

Our algorithm will search $G$ for cells. As $G$ has $\Omega(n^2)$ cells, we cannot afford to maintain $G$ explicitly. Instead, we use a balanced binary search tree $T_v$ to store the vertical lines of $L_v$ ordered by their $x$-coordinates and use a tree $T_h$ to store the horizontal lines of $L_h$ ordered by their $y$-coordinates. We use $T_v$ and $T_h$ to maintain $G$ implicitly, i.e., operations on $G$ will be done using these two trees. 

\paragraph{Computing the set $\boldsymbol{\calC}$.} We are now ready to compute the conforming coverage set $\calC$ by using $G$. 

For each point $p\in P$, we find the cell of $G$ that contains $p$ by doing binary search on the lines of $L_v$ and on the lines of $L_h$. This can be done in $O(\log n)$ time using the two trees $T_v$ and $T_h$. As such, in $O(n\log n)$ time, we can find all $O(n)$ non-empty cells of $G$. For each such non-empty cell $C$, the points of $P(C)$ are also computed. 

For each point $p\in P$, consider the cell $C_p$ of $G$ that contains $p$. Since $p$ is both in a horizontal point-zone and a vertical point-zone, $C_p$ is regular cell. Further, 
due to the above properties (4) of the vertical and horizontal point-zones, $\boxplus_{C_p}$ is in a horizontal point-zone and also in a vertical point-zone, and thus all cells of $\boxplus_{C_p}$ are regular cells. 
Given $C_p$, we can find all $O(1)$ cells of $\boxplus_{C_p}$ in $O(\log n)$ using the trees $T_v$ and $T_h$. For each cell $C\in \boxplus_{C_p}$, we add $C_p$ to its neighboring set $N(C)$. It is not difficult to see that $|N(C)|=O(1)$ since $C$ can only be in $\boxplus_{C'}$ for $O(1)$ cells $C'\in G$. 
In this way, due to the above key property of $\boxplus_{C}$, we have the following observation. 

\begin{observation}\label{obser:cellsubset}
For any point $q$, if $C_q$ is the cell of $G$ containing $q$, then $P(D_q)\subseteq \bigcup_{C\in N(C_q)}P(C)$.     
\end{observation}

We define the conforming coverage set $\calC$ as the set of all cells of $\boxplus_{C}$ for all non-empty cells $C\in G$. The above computes $\calC$, along with $P(C)$ and $N(C)$ for all cells $C\in \calC$, in $O(n\log n)$ time. Since $\boxplus_{C}$ has $O(1)$ cells, $|\calC|=O(n)$. We argue that all conditions in Definition~\ref{def:cell} hold for $\calC$. Indeed, the first three conditions follow directly from the definition of $\calC$. The fourth condition holds due to Observation~\ref{obser:cellsubset}. For the fifth condition, consider any point $q$ such that the cell $C_q$ of $G$ containing $q$ is not in $\calC$. Then, $C_q$ is not in $\boxplus_{C}$ for any non-empty cell $C$, and thus $q\not\in\boxplus_{C}$. By the key property of $\boxplus_{C}$, $D_q\cap C=\emptyset$. Consequently, $D_q\cap P=\emptyset$. Therefore, the fifth condition also holds. 

\paragraph{Answering queries.}
Given a query point $q$, we wish to determine whether $q$ is in a cell $C$ of $\calC$, and if so, return $C$. To this end, we need to store cells of $\calC$ in a data structure. 
For each cell $C\in G$, we use its bottom left corner point as its ``id'' or ``representative point''. 
We sort all cells of $\calC$ by the lexicographical order of the coordinates of their representative points, i.e., for any two points $q_1$ and $q_2$ in the plane, we let $q_1$ be smaller than $q_2$ if $x(q_1)<x(q_2)$, or if $x(q_1)=x(q_2)$ and $y(q_1)<y(q_2)$. We use a balanced binary search tree $T$ to store cells of $\calC$ following the above order. Furthermore, to facilitate point deletions, for each non-empty cell $C$ of $\calC$, we store all points of $P(C)$ in a balanced binary search tree $T_C$ following the lexicographical order of the coordinates of the points. This finishes constructing our data structure, which takes $O(n\log n)$ time and $O(n)$ space. 

Given a query point $q$, we can answer the query in $O(\log n)$ time as follows. First, we find the cell $C_q$ of $G$ that contains $q$ using the two trees $T_h$ and $T_v$. Then, we determine whether $C_q\in \calC$ by searching the representative point of $C_q$ in $T$. 

\subsection{Handling insertions}

Suppose that we want to insert a new point $p^*$ to $P$. Roughly speaking, the goal of our insertion algorithm is to update our data structure so that it would be the same as what we had built it on $P\cup \{p^*\}$. This means we may need to update the grid $G$ in order to include $p^*$ in both a horizontal point-zone and a vertical point zone. As will be seen later, we may need to insert more vertical partition lines inside some vertical gap-zones (but no vertical partition line will be inserted inside any vertical point-zone). The effect is that a vertical point-zone is expanded (and thus a gap-zone is shrunk), a new vertical point-zone is created inside a gap-zone (and the gap-zone is split into two smaller gap-zones plus a point-zone in the middle), or two adjacent vertical point-zones and the gap-zone between them is merged into a new larger point-zone. Similarly, we may need to insert more horizontal partition lines inside some horizontal gap-zones (but no horizontal partition line will be inserted inside any horizontal point-zone). As such, some gap cells of $G$ may be divided into smaller cells but regular cells will never be changed. 

In fact, since the point $p^*$ is given ``online'', we may not be able to build the same data structure as we knew $P\cup \{p^*\}$. One issue is that the width of a vertical gap-zone may not be a multiple of $1/2$. If we need to include the entire gap-zone into a point-zone due to insertions, then it is not always possible to guarantee that the width of each regular column is equal to $1/2$. Similar issues also happen to horizontal gap-zones. To address these issues, we allow regular cells to have side lengths smaller than $1/2$. As some regular cells might be too small, to ensure that the key property of $\boxplus_C$ still holds, our algorithm maintains an invariant that there are at least seven regular columns of width $1/2$ between every two columns of width smaller than $1/2$ (called ``narrow columns''); similarly, there are at least seven regular rows of heights $1/2$ between every two ``narrow rows''. This invariant guarantees that the key property of $\boxplus_C$ still holds. In what follows, we present the details of the insertion algorithm. 

First, using $T_h$ and $T_v$, we find the cell $C^*$ that contains $p^*$. We determine whether $C^*\in \calC$ by searching the tree $T$. All these takes $O(\log n)$ time. 

According to the definition of $G$ and our algorithm for constructing the horizontal and vertical point-zones, each vertical point-zone has at least seven columns and each horizontal point-zone has at least seven rows. Since $\boxplus_{C^*}$ intersects exactly seven rows (resp., columns) of $G$, the interior of $\boxplus_{C^*}$ can intersect at most two vertical point-zone bounding lines, and if it intersects two such lines, then both lines are the bounding lines of the same vertical gap-zone. Similarly, the interior of $\boxplus_{C^*}$ can intersect at most two horizontal point-zone bounding lines, and if it intersects two such lines, then both lines are the bounding lines of the same horizontal gap-zone.
Depending on whether the interior of $\boxplus_{C^*}$ intersects any vertical (resp., horizontal) point-zone bounding line, there are four cases: (1) it does not intersect any point-zone bounding line; (2) it intersects a vertical point-zone bounding line but does not intersect any horizontal point-zone bounding line; (3) it intersects a horizontal point-zone bounding line but does not intersect any vertical point-zone bounding line; (4) it intersects both a vertical point-zone bounding line and a horizontal point-zone bounding line. We discuss the four cases below. The first case is relatively easy to handle. Our main effort focuses on the second case. The third case is symmetric to the second one and thus can be handled similarly. For the fourth case, the algorithm essentially first runs the second case algorithm and then runs the third case algorithm. Note that which of these four cases occurs can be determined in $O(\log n)$ time using trees $T_v$ and $T_h$.

\paragraph{Case (1): The interior of $\boxplus_{C^*}$ does not intersect any point-zone bounding line.}
In this case, $\boxplus_{C^*}$ must be inside a vertical point-zone and also inside a horizontal point-zone. Indeed, this is the case because $C^*$ cannot be a gap cell (since a gap cell must have at least one edge either on a horizontal point-zone bounding line or a vertical point-zone bounding line, which means that the interior of $\boxplus_{C^*}$ must intersect a point-zone bounding line). Therefore, all cells of $\boxplus_{C^*}$ are regular cells. Depending on whether $C^*$ contains a point of $P$, there are further two cases. 

\begin{itemize}
    \item If $C^*$ contains a point of $P$, i.e., $C^*$ is a non-empty cell, then all cells of $\boxplus_{C^*}$ are in $\calC$ and thus are already stored in $T$, and $T_{C^*}$ stores all points of $P(C^*)$. We simply insert $p^*$ to $T_{C^*}$. This finishes the insertion. The total time of the insertion algorithm is $O(\log n)$. 
    
    \item If $C^*$ does not contain any point of $P$, then for each cell $C\in \boxplus_{C^*}$ (note that we can find all $O(1)$ cells of $\boxplus_{C^*}$ in $O(\log n)$ time using $T_v$ and $T_h$), we add $C^*$ to $N(C)$ and insert $C$ to $T$ if $C$ is not already there. 
    In addition, we initiate a tree $T_{C^*}$ to store $p^*$. This finishes the insertion. Since $\boxplus_{C^*}$ has $O(1)$ cells, the total time of the insertion algorithm is $O(\log n)$.  
\end{itemize}
        
\paragraph{Case (2): The interior of $\boxplus_{C^*}$ intersects a vertical point-zone bounding line but does not intersect any horizontal point-zone bounding line.} 

In this case, depending on whether $C^*$ is a regular cell, there are further two cases. 

We first discuss the case where $C^*$ is a regular cell. Let $Z^*$ denote the vertical point-zone containing $C^*$. 
If the interior of $\boxplus_{C^*}$ intersects only one vertical point-zone bounding line, then let $\ell^*$ denote that line. If it intersects two such lines, then both lines are on the same side of $Z^*$ because $Z^*$ contains at least seven regular columns; in this case, let $\ell^*$ refer to the one of the two bounding lines closer to $C^*$. Without loss of generality, we assume that $C^*$ is to the left of $\ell^*$ (see Fig.~\ref{fig:pointzone}).

    \begin{figure}[t]
        \centering
        \includegraphics[height=2.5in]{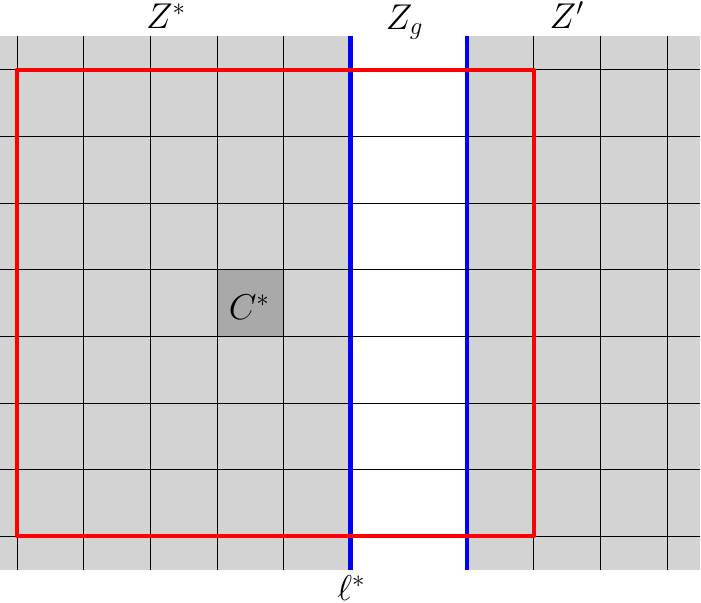}
        \caption{Illustrating the case where $C^*$ is a regular cell. The red box is $\boxplus_{C^*}$. The two blue vertical lines are vertical point-zone bounding lines; the left one is $\ell^*$. All grey cells are in vertical point-zones  while cells are in vertical gap zones.}
        \label{fig:pointzone}
        \vspace{-0.15in}
    \end{figure}

Observe that $\ell^*$ must be the right bounding line of $Z^*$. Since the interior of $\boxplus_{C^*}$ intersects $\ell^*$, $Z^*$ is not large enough to incorporate $p^*$ because $\ell^*$ is too close to $p^*$. In other words, if we run our sweeping algorithm on $P\cup \{p^*\}$ to compute the vertical point-zones, then $\ell^*$ would not be the right bounding line of $Z^*$. Therefore, to incorporate $p^*$, we need to move $\ell^*$ to the right. Let $Z_g$ be the right 
vertical gap-zone adjacent to $Z^*$. We wish to move $\ell^*$ rightwards towards $Z_g$ so that the distance between the new $\ell^*$ and $C^*$ is exactly $3/2$ (we can then add at most three non-boundary vertical partition lines to create new regular columns between $C^*$ and the new $\ell^*$; afterwards, the interior of $\boxplus_{C^*}$ in the new grid $G$ does not intersect any point-zone bounding line anymore). However, $Z_g$ may not have enough space to make this possible, i.e., the move of $\ell^*$ may cross the right boundary of $Z_g$. Depending on whether this move is possible, there are further two cases. 

\begin{itemize}
    \item If the move of $\ell^*$ as above does not cross over the right boundary of $Z_g$, then we move $\ell^*$ and add the  non-boundary vertical partition lines as discussed above (i.e., we insert these at most three non-boundary vertical lines to the tree $T_v$). After that, the interior of the new $\boxplus_{C^*}$ in the new grid $G$ does not intersect any point-zone bounding line. For each cell $C\in \boxplus_{C^*}$, we add $C^*$ to $N(C)$ and insert $C$ to $T$ if $C$ is not already there. Also, we initiate a tree $T_{C^*}$ to store $p^*$. This finishes the insertion. Since $\boxplus_{C^*}$ has $O(1)$ cells, the total time is $O(\log n)$. 
    
    \item If the move of $\ell^*$ as above crosses over the right boundary of $Z_g$, then we stop moving $\ell^*$ at a position $x$ such that the distance between $x$ and the old $\ell^*$ is a multiple of $1/2$, and the distance between $x$ and the right boundary of $Z_g$ is less than $1/2$. In fact, the distance between $x$ and the old $\ell^*$ can be $0$, $1/2$, or $1$. Such a position $x$ can be found in $O(1)$ time. We again add non-boundary vertical partition lines between the old $\ell^*$ and the new $\ell^*$ so that each new column has width $1/2$. Also, the width of the new gap zone $Z_g$ becomes less than $1/2$. Let $Z'$ be the right neighboring point-zone of $Z_g$. We merge $Z^*$, $Z_g$, and $Z'$ together to create a single vertical point-zone. The merge is done by simply marking $\ell^*$ a non-boundary vertical partition line. In the new point-zone, $Z_g$ becomes a regular column whose width is smaller than $1/2$, and we call it a ``narrow column'' (see Fig.~\ref{fig:pointzone10}).

    \begin{figure}[t]
        \centering
        \includegraphics[height=2.5in]{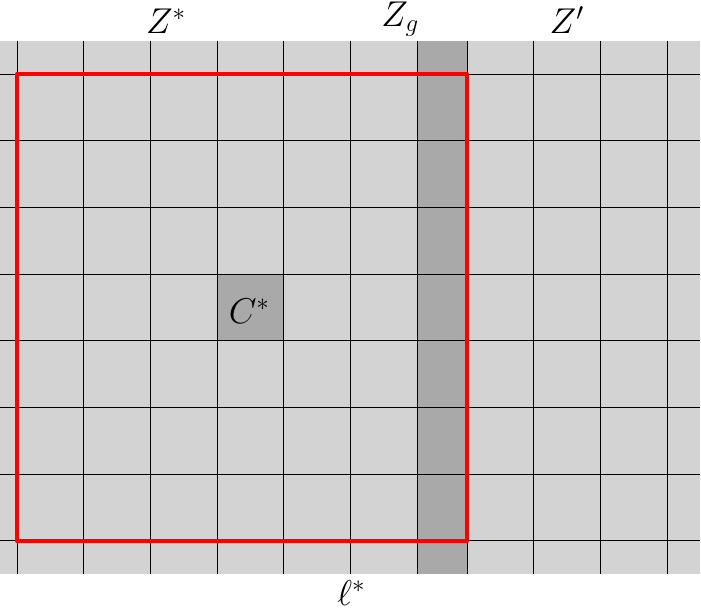}
        \caption{Illustrating the new grid after the update on Figure~\ref{fig:pointzone}. The dark grey column is a narrow column. The red box is the new $\boxplus_{C^*}$.}
        \label{fig:pointzone10}
        \vspace{-0.15in}
    \end{figure}
    
    Our algorithm maintains an invariant that for each bounding line $\ell$ of any vertical point-zone $Z$, none of the seven columns of $Z$ closest to $\ell$ is a narrow column. It is not difficult to see that the above merge operation maintains this invariant. The invariant implies that there are at least seven regular columns of width $1/2$ between every two narrow columns. This further leads to the following observation: For any regular cell $C\in G$, $\boxplus_{C}$ intersects at most one narrow column. Similarly, our algorithm maintains the invariant that $\boxplus_{C}$ of every regular cell $C$ intersects at most one ``narrow row''. This guarantees that the key property of $\boxplus_C$ still holds even if we allow narrow columns and narrow rows in $G$. 

    After the above merge operation, the interior of the new $\boxplus_{C^*}$ in the new grid $G$ does not intersect any point-zone bounding line (see Fig.~\ref{fig:pointzone10}). For each cell $C\in \boxplus_{C^*}$, we add $C^*$ to $N(C)$ and insert $C$ to $T$ if $C$ is not already there. In addition, we initiate a tree $T_{C^*}$ to store $p^*$. This finishes the insertion, which takes $O(\log n)$ time in total. 
\end{itemize}

Next, we discuss the case where $C^*$ is not a regular cell. In this case, $C^*$ is contained in a gap-zone $Z^*$ whose left bounding line $\ell_1$ contains the left edge of $C^*$ and whose right bounding line $\ell_2$ contains the right edge of $C^*$. Hence, the interior of $\boxplus_{C}$ intersects two vertical point-zone bounding lines, i.e., $\ell_1$ and $\ell_2$ (see Fig.~\ref{fig:pointzone20}). Without loss of generality, we assume that $p^*$ is closer to $\ell_1$ than to $\ell_2$, i.e., $|x(p^*)-x(\ell_1)|\leq |x(p^*)-x(\ell_2)|$. Let $Z_1$ (resp., $Z_2$) be the left (resp., right) neighboring vertical point-zone of $Z^*$. 
Roughly speaking, our algorithm for this case works as follows. If $p^*$ is far enough from $\ell_1$, then we will create a new vertical point-zone inside $Z^*$ that contains $p^*$ (and $Z^*$ is thus split into two smaller gap-zones and one vertical point-zone between them). Otherwise, we expand $\ell_1$ to the right to make the new $Z_1$ contain $p^*$, and if $p^*$ is also close to $\ell_2$, then $Z_1$, $Z^*$, and $Z_2$ are merged into a new point-zone. The details are elaborated below.

    \begin{figure}[t]
        \centering
        \includegraphics[height=2.5in]{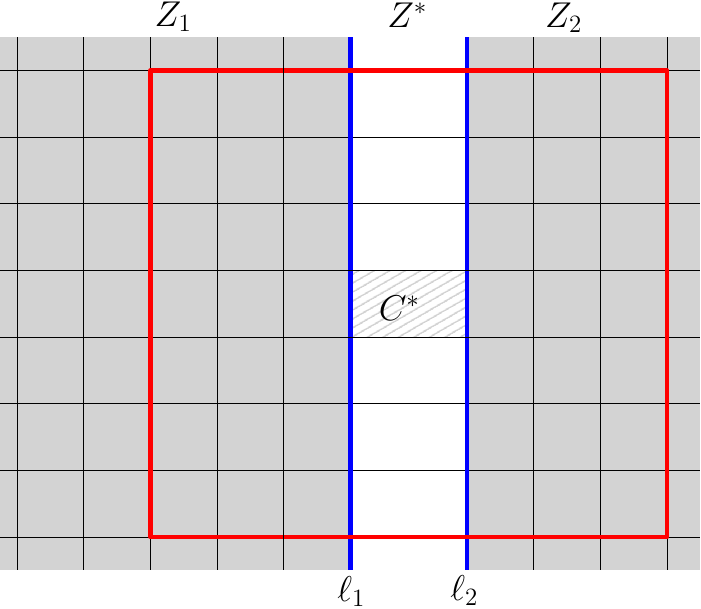}
        \caption{Illustrating the case where $C^*$ is in a vertical gap-zone. The red box is $\boxplus_{C^*}$. All grey cells are in vertical point-zones while white cells are in vertical gap zones. $\ell_1$ and $\ell_2$ are respectively the left and right bounding lines of the vertical gap-zone containing $C^*$.}
        \label{fig:pointzone20}
        \vspace{-0.15in}
    \end{figure}

\begin{itemize}    
\item 
If $|x(p^*)-x(\ell_1)|>7/4$, then since $|x(p^*)-x(\ell_1)|\leq |x(p^*)-x(\ell_2)|$, we create a new point-zone $Z'$ with left bounding line $\ell_1'$ at $x(p^*)-7/4$ and right bounding line $\ell_2'$ at $x(p^*)+7/4$. We then add six vertical regular partition lines in $Z'$ to partition it into seven columns of width $1/2$. We add these lines to $T_v$. We still use $C^*$ to refer to the new cell containing $p^*$ in the new $G$. Notice that the new $\boxplus_{C^*}$ now does not intersect any point-zone bounding line. Therefore, all cells of $\boxplus_{C^*}$ are now regular cells. For each cell $C\in \boxplus_{C^*}$, we add $C^*$ to $N(C)$ and insert $C$ to $T$ if $C$ is not already there. In addition, we initiate a tree $T_{C^*}$ for $C^*$ to store $p^*$. This finishes the insertion, which takes $O(\log n)$ time in total. 

\item 
If $|x(p^*)-x(\ell_1)|\leq 7/4$, then we expand the point-zone $Z_1$ by moving its right bounding line $\ell_1$ rightwards to contain $p^*$. Specifically, we move $\ell_1$ rightwards and at every distance $1/2$ we add a non-boundary vertical partition line to $T_v$. After passing $p^*$, we continue the process until four vertical partition lines are added after $p^*$ (and $\ell_1$ will be at the position of the last partition line). Assuming that $\ell_1$ has not passed over $\ell_2$, we proceed as follows. Since $|x(p^*)-x(\ell_1)|\leq 7/4$, the above process has added at most seven vertical lines to $T_v$, in $O(\log n)$ time. We still use $C^*$ to refer to the new cell that contains $p^*$ in the new $G$. Notice that the new $\boxplus_{C^*}$ is now inside the new vertical point-zone $Z_1$ and the gap-zone $Z^*$ is shrunk. Hence, $\boxplus_{C^*}$ now does not intersect any point-zone bounding line. Therefore, all cells of $\boxplus_{C^*}$ are now regular cells. For each cell $C\in \boxplus_{C^*}$, we add $C^*$ to $N(C)$ and insert $C$ to $T$ if it is not already there. In addition, we initiate a tree $T_{C^*}$ for $C^*$ to store $p^*$. This finishes the insertion, which takes $O(\log n)$ time in total. 

The above assumes that the moving of $\ell_1$ does not cross $\ell_2$. If it does, then the gap-zone $Z^*$ is too small. In this case, the moving of $\ell_1$ should stop as soon as the distance between the current partition line and $\ell_2$ is smaller than $1/2$. Then, we let $\ell_1$ be the current partition line. The new $Z^*$ refers to the region between the new $\ell_1$ and $\ell_2$, which is a narrow column. Note that the above added at most six vertical lines to $L_v$. We now merge $Z_1$, $Z^*$, and $Z_2$ into a single point-zone, simply by marking $\ell_1$ and $\ell_2$ as non-boundary vertical partition lines. We again use $C^*$ to refer to the cell containing $p^*$ in the new $G$. The new $\boxplus_{C^*}$ now does not intersect any point-zone bounding line. Therefore, all cells of $\boxplus_{C^*}$ are now regular cells. For each cell $C\in \boxplus_{C^*}$, we add $C^*$ to $N(C)$ and insert $C$ to $T$ if it is not already there. In addition, we initiate a tree $T_{C^*}$ for $C^*$ to store $p^*$. This finishes the insertion, which takes $O(\log n)$ time in total. 
\end{itemize}
       
\paragraph{Case (3): The interior of $\boxplus_{C^*}$ intersects a horizontal point-zone bounding line but does not intersect any vertical point-zone bounding line.}

This case is symmetric to the above second case. In this case, we may need to insert horizontal partition lines to $T_h$ to  expand a horizontal point-zone upwards or downwards, or merge two horizontal point-zones and a gap-zone, which may produce a ``narrow row''. As discussed above, the algorithm maintains an invariant that there are at least seven rows of width $1/2$ between every two narrow rows. 

\paragraph{Case (4): The interior of $\boxplus_{C^*}$ intersects a vertical point-zone bounding line and also a horizontal point-zone bounding line.}
In this case, we first add vertical partition lines to $T_v$ as in the above second case and then add horizontal partition lines to $T_h$ as in the above third case. Afterwords, let $C^*$ refer to the cell that contains $p^*$ in the new grid $G$. The new $\boxplus_{C^*}$ is now inside a vertical point-zone and also inside a horizontal point-zone. Therefore, all cells of $\boxplus_{C^*}$ are now regular cells. For each cell $C\in \boxplus_{C^*}$, we add $C^*$ to $N(C)$ and insert $C$ to $T$ if it is not already there. In addition, we initiate a tree $T_{C^*}$ for $C^*$ to store $p^*$. This finishes the insertion. The total time is $O(\log n)$. 

\subsection{Dealing with deletions}
Suppose we wish to delete a point $p^*$ from $P$. We first find the cell $C\in \calC$ that contains $p^*$. This can be done in $O(\log n)$ time by first finding $C$ in $G$ using the two trees $T_v$ and $T_h$, and then searching $C$ in $T$. Recall that all points of $P(C)$ are stored in the tree $T_C$. We simply remove $p^*$ from $T_C$. This finishes the deletion, which takes $O(\log n)$ time. In addition, to make sure the size of $\calC$ is $O(|P|)$, after $n/2$ deletions, we recompute $\calC$ and the data structure based on the current set $P$, which takes $O(n\log n)$ time. Hence, the amortized deletion time is $O(\log n)$.

\section{Dynamic range reporting}
\label{sec:dynamicreport}

Let $P$ be a set of points in the plane. We wish to maintain a data structure for $P$ to answer unit-disk range reporting queries subject to point insertions and deletions of $P$. Let $n$ denote the size of the current set $P$. 

Using Lemma~\ref{lem:dynamiccell}, we maintain a conforming coverage set $\calC$ of $O(n)$ cells for $P$. To insert a point $p$ to $P$, our insertion algorithm for Lemma~\ref{lem:dynamiccell} boils down to inserting $p$ to $P(C)$ for a cell $C\in \calC$ that contains $p$. To delete a point $p$ from $P$, our deletion algorithm for Lemma~\ref{lem:dynamiccell} boils down to deleting $p$ from $P(C)$ for a cell $C\in \calC$ that contains $p$. 

Consider a query unit disk $D_q$ whose center is $q$. If $q$ is not in a cell of $\calC$, then by Definition~\ref{def:cell}(5), $P(D_q)=\emptyset$ and thus we simply return null. Otherwise, 
to report $P(D_q)$, it suffices to report $P(C') \cap D_q$ for all cells $C' \in N(C)$. In the case of $C' = C$,
since the distance between two points in $C$ is at most 1 by Definition~\ref{def:cell}(1), we can simply
report all points of $P(C)$. If $C' \neq C$, $C$ and $C'$ are separated by an axis-parallel line by Definition~\ref{def:cell}(3). Without loss of generality, we assume that $C$ and $C'$ are separated by a horizontal line $\ell$ with $C'$ above $\ell$ and $C$ below $\ell$. As $q\in C$, $q$ is below $\ell$, i.e., $q$ is separated from $C'$ by $\ell$. Our goal is to report points of $P(C')\cap D_q$. Due to the point updates of $P(C')$, 
our problem is reduced to the following subproblem, called {\em dynamic line-separable UDRR problem}. 

\begin{problem}(Dynamic line-separable UDRR)
\label{problem:dynamic-LS-UDRR}
For a set $Q$ of $m$ points above a horizontal line $\ell$, 
maintain $Q$ in a data structure to support the following operations. 
(1) Insertion: insert a point to $Q$; (2) deletion: delete a point from $Q$; (3) unit-disk range reporting query: given a point $q$ below $\ell$, report the points of $Q$ in the unit disk $D_q$.
\end{problem}

We have the following Lemma~\ref{lem:dynamiclineUDRR} for the dynamic line-separable UDRR problem. 

\begin{lemma}\label{lem:dynamiclineUDRR}
For the dynamic line-separable UDRR problem, we can have a data structure of $O(m\log m)$ space to maintain $Q$ to support insertions in $O(\log^{3+\epsilon} m)$ amortized time, deletions in $O(\log^{5+\epsilon} m)$ amortized time, and unit-disk range reporting queries in $O(k+\log m)$ time, where $k$ is the output size, $\epsilon$ is an arbitrarily small positive constant, and $m$ is the size of the current set $Q$. 
\end{lemma}

We will prove Lemma~\ref{lem:dynamiclineUDRR} in Section~\ref{sec:lemdynamiclineUDRR}. With Lemmas~\ref{lem:dynamiccell} and \ref{lem:dynamiclineUDRR}, we can obtain the following main result for our original dynamic UDRR problem.

\begin{theorem}
    \label{theo:dynamicudrr}
    We can maintain a set $P$ of points in the plane in a data structure of $O(n\log n)$ space to support insertions in $O(\log^{3+\epsilon} n)$ amortized time, deletions in $O(\log^{5+\epsilon} n)$ amortized time, and unit-disk range reporting queries in $O(k+\log n)$ time, where $k$ is the output size, $\epsilon$ is an arbitrarily small positive constant, and $n$ is the size of the current set $P$. 
\end{theorem}
\begin{proof}
We build the data structure $\calD$ in Lemma~\ref{lem:dynamiccell} to maintain a conforming coverage set $\calC$ of $O(n)$ cells for $P$. For each cell $C\in \calC$ that contains at least one point of $P$, we maintain a data structure $\calD_e(C)$ for $P(C)$ with respect to the supporting line of each edge $e$ of $C$. 
Since the space of each $\calD_e(C)$ is $O(|P(C)| \log |P(C)|)$, each cell of $\calC$ has four edges, and $\sum_{C\in \calC}|P(C)|=n$, the total space of the overall data structure is $O(n \log n)$. 

\paragraph{Insertions.}
To insert a point $p$ to $P$, we first update $\calD$ by Lemma~\ref{lem:dynamiccell}, which takes $O(\log n)$ worst-case time. The insertion algorithm of Lemma~\ref{lem:dynamiccell} eventually inserts $p$ to $P(C)$ for a cell $C\in \calC$ that contains $p$. We insert $p$ to $\calD_e(C)$ for each edge $e$ of $C$, which takes $O(\log^{3+\epsilon}n)$ amortized time by Lemma~\ref{lem:dynamiclineUDRR}. As such, each insertion takes $O(\log^{3+\epsilon}n)$ amortized time. 

\paragraph{Deletions.}
To delete a point $p$ from $P$, we first update $\calD$ by Lemma~\ref{lem:dynamiccell}, which takes $O(\log n)$ amortized time. The deletion algorithm of Lemma~\ref{lem:dynamiccell} eventually deletes $p$ from $P(C)$ for a cell $C\in \calC$ that contains $p$. We delete $p$ from $\calD_e(C)$ for each edge $e$ of $C$, which takes $O(\log^{5+\epsilon}n)$ amortized time by Lemma~\ref{lem:dynamiclineUDRR}. As such, each deletion takes $O(\log^{5+\epsilon}n)$ amortized time. 

\paragraph{Queries.}
Given a query unit disk $D_q$ with center $q$, we first check whether $q$ is in a cell of $\calC$, and if so, find such a cell; this takes $O(\log n)$ time by Lemma~\ref{lem:dynamiccell}. If no cell of $\calC$ contains $q$, then $P\cap D_q=\emptyset$ and we simply return null. Otherwise, let $C$ be the cell of $\calC$ that contains $q$. We first report all points of $P(C)$. Next, for each $C'\in N(C)$, by Definition~\ref{def:cell}(3), $C$ and $C'$ are separated by an axis-parallel line $\ell$. Since each edge of $C$ and $C'$ is axis-parallel, $C'$ must have an edge $e$ whose supporting line is parallel to $\ell$ and separates $C$ and $C'$. Using  $\calD_e(C')$, we report all points of $P(C')$ inside $D_q$. As $|N(C)|=O(1)$, the total query time is $O(\log n + k)$ by Lemma~\ref{lem:dynamiclineUDRR}.
\end{proof}

\section{Proving Lemma~\ref{lem:dynamiclineUDRR}: Dynamic line-separable UDRR}
\label{sec:lemdynamiclineUDRR}

We now prove Lemma~\ref{lem:dynamiclineUDRR}. For notational convenience, instead of $m$, we use $n$ to denote the size of $Q$. 

Consider a query unit disk $D_q$ with center $q$ below $\ell$. The goal is to report $Q(D_q)$.
Observe that a point $p \in Q$ is in $D_q$ if and only if $q$ is in the unit
disk $D_p$. The portion of $\partial D_p$ below $\ell$ is a circular arc, denoted by $\gamma_p$.
Since $p$ is above $\ell$, $\gamma_p$ is on the lower half circle of $\partial D_p$ and
thus is $x$-monotone. As such, $p$ is in $D_q$ if and only if $q$ is above the arc $\gamma_p$. 
Define $\Gamma$ to be the set of arcs $\gamma_p$ for all points $p\in Q$. 
Therefore, reporting the points of $Q$ in $D_q$ becomes reporting the arcs of $\Gamma$ that are below $q$, which we call {\em arcs reporting queries}. 

In what follows, an arc of $\Gamma$ always refers to the portion below $\ell$ of a unit circle with center above $\ell$. 
Our problem thus becomes dynamically maintaining a set $\calA$ of arcs to report the arcs of $\calA$ below a query point $q$. The arcs reporting queries can be reduced to the following {\em $k$-lowest-arcs queries}: Given a query vertical line $\ell^*$ and a number $k\geq 1$, report the $k$ lowest arcs of $\calA$ intersecting $\ell^*$. We have the following observation, which follows the proof of Chan~\cite{ref:ChanRa00} for lines. 
\begin{observation}{\em\cite{ref:ChanRa00}}\label{obser:karcquery}
Suppose that we can answer each $k$-lowest-arcs query in $O(\log n+k)$ time. Then, the arcs of $\calA$ below a query point $q$ can be reported in $O(\log n+ k)$ time, where $k$ is the output size. 
\end{observation}
\begin{proof}
Given a query point $q$, let $k$ be the number of arcs of $\calA$ below $q$, which is not known in advance. Let $\ell^*$ be the vertical line through $q$. Let $k_i=2^i\log n$. We apply $k_i$-lowest arc queries for $i=0,1,\ldots$ until the algorithm reports an arc whose intersection with $\ell^*$ is higher than $q$. Then, among all $k_i$ arcs reported by the $k_i$-lowest-arcs queries, we return all arcs whose intersections are below $q$. The correctness is obvious. The runtime is on the order of $\log n+\sum_{k_{i-1}<k}(\log n+k_i)=O(\log n+k)$. 
\end{proof}

In light of the above observation, we now focus on the $k$-lowest-arcs queries. 
We adapt a technique for a similar problem on lines (which is the dual problem of the dynamic halfplane range reporting problem): Dynamically maintain a set of lines (subject to insertions and deletions) to report the $k$-lowest lines at a query vertical line. For this problem, Chan~\cite{ref:ChanTh12} gave a data structure of $O(n\log n)$ space that supports 
$O(\log^{6+\epsilon} n)$ amortized update time
and $O(k+\log n)$ query time. 
De Berg and Staals~\cite{ref:BergDy23} improved the result of \cite{ref:ChanTh12} for dynamically maintaining a set of planes in 3D. They gave a data structure of $O(n\log n)$ space that supports $O(\log^{3+\epsilon} n)$ amortized insertion time, 
$O(\log^{5+\epsilon} n)$ amortized deletion time, and $O(\log^2 n/\log\log n+k)$ query time. Their approach is based
on the techniques for dynamically maintaining planes for answering lowest point queries~\cite{ref:ChanA10,ref:ChanDy20,ref:KaplanDy20} and these techniques in turn replies on computing shallow cuttings on the planes in 3D~\cite{ref:ChanOp16}. 
In the following, we will extend these techniques to the arcs of $\calA$ and prove the following result. 
\begin{lemma}\label{lem:karcquery}
For the set $\calA$ of arcs, we can have a data structure of $O(n\log n)$ space to support insertions in $O(\log^{3+\epsilon} n)$ amortized time, deletions in $O(\log^{5+\epsilon} n)$ amortized time, and $k$-lowest-arcs queries in $O(k+\log n)$ time, where $n$ is the size of the current set $\calA$. 
\end{lemma}

Combining Lemma~\ref{lem:karcquery} and Observation~\ref{obser:karcquery} immediately leads to Lemma~\ref{lem:dynamiclineUDRR}. 

In what follows, we first develop a shallow cutting algorithm for arcs of $\Gamma$ in Section~\ref{sec:defshallow} and then using the algorithm to prove Lemma~\ref{lem:karcquery} in Section~\ref{sec:karcquery}.

\subsection{Shallow cuttings}
\label{sec:defshallow}

Without loss of generality, we assume that $\ell$ is the $x$-axis. 
Let $\bbR^-$ (resp., $\bbR^+$) be the half-plane below (resp., above) $\ell$. 
Note that each arc of $\calA$ is $x$-monotone, every arc has both endpoints on $\ell$, and every two arcs cross each other at most once.

We use {\em $\bbR^+$-constrained unit disk} to refer to a unit disk with center in $\bbR^+$ and use {\em $\bbR^+$-constrained arc} to refer to a portion of the arc $C\cap \bbR^-$ for a unit circle $C$ with center in $\bbR^+$.
For any point $q\in \bbR^-$, let $\rho(q)$ to denote the vertical downward ray from $q$. We say that an arc $\gamma$ of $\calA$ is {\em below} $q$ if it intersects $\rho(q)$. As the center of $\gamma$ is in $\bbR^+$ and $\rho_q\in \bbR^-$, $\gamma$ intersects $\rho_q$ at most once. 

For a parameter $r\leq n$ and a region $R$ of the plane, a {\em $(1/r)$-cutting covering $R$} for the arcs of $\calA$ is a set of interior-disjoint cells such that the union of all cells covers $R$ and each cell intersects at most $n/r$ arcs of $\calA$. For each cell $\Delta$, its {\em conflict list} $\calA_{\Delta}$ is the set of arcs of $\calA$ that intersect $\Delta$. The size of the cutting is the number of its cells.

For a point $p\in \bbR^-$, the {\em level} of $p$ in $\calA$ is the number of arcs of $\calA$ below $p$. For any integer $k\in [1,n]$, the {\em $(\leq k)$-level} of $\calA$, denoted by $L_{\leq k}(\calA)$, is defined as the region consisting of all points of $\bbR^-$ with level at most $k$. 
Given parameters $r, k\in [1,n]$, a {\em $k$-shallow $(1/r)$-cutting} is a $(1/r)$-cutting for $\calA$ that covers $L_{\leq k}(\calA)$. 

We use {\em pseudo-trapezoid} to refer to a region that has two vertical line segments as left and right edges, an $\bbR^+$-constrained arc or a line segment on $\ell$ as a top edge, and an $\bbR^+$-constrained arc as a bottom edge. In particular, if a pseudo-trapezoid does not have a bottom edge, i.e., the bottom side is unbounded, then we call it a {\em bottom-open} pseudo-trapezoid.

We say that a shallow cutting is in the {\em bottom-open pseudo-trapezoid form} if every cell of it is a bottom-open pseudo-trapezoid. 
Our main result about the shallow cuttings for $\calA$ is given in the following theorem. 

\begin{theorem}\label{theo:shallowcut}
There exist constants $B$, $C$, and $C'$, such that for a parameter $k\in [1,n]$, we can compute a $(B^ik)$-shallow $(CB^ik/n)$-cutting of size at most $C'\frac{n}{B^ik}$ in the bottom-open pseudo-trapezoid form, along with conflict lists of all its cells, for all $i=0,1,\ldots,\log_B\frac{n}{k}$, in $O(n\log \frac{n}{k})$ total time. In particular, we can compute a $k$-shallow $(Ck/n)$-cutting of size $O(n/k)$, along with its conflict lists, in $O(n\log \frac{n}{k})$ time.   
\end{theorem}

Since the proof of Theorem~\ref{theo:shallowcut} is technical and lengthy (and is one of our main results in this paper), we devote the entire Section~\ref{sec:algoshallow} to it.

\subsection{Proving Lemma~\ref{lem:karcquery}}
\label{sec:karcquery}
We now prove Lemma~\ref{lem:karcquery}. 
With the shallow cutting algorithm in Theorem~\ref{theo:shallowcut}, we generalize the techniques of \cite{ref:BergDy23,ref:ChanTh12} for lines to the arcs of $\calA$. 
We first give two deletion-only data structures, which will be needed in our fully dynamic data structure for Lemma~\ref{lem:karcquery}. 

\subsubsection{Deletion-only data structure}

Our first deletion-only data structure is given in Lemma~\ref{lem:delonly}, whose proof in turn relies on Lemma~\ref{lem:delonlypoint}.

\begin{lemma}\label{lem:delonlypoint}
There is a data structure of $O(n)$ size to maintain a set $\calA$ of $n$ arcs to support $O(\log n)$ amortized time deletions, such that given a query point $q$, the arcs of $\calA$ below $q$ can be computed in $O(\sqrt{n}\log^{O(1)} n +k)$ time, where $k$ is the output size and $n$ is the size of the current set $\Gamma$. If a set $\Gamma$ of $n$ arcs is given initially, 
the data structure can be built in $O(n\log n)$ time. 
\end{lemma}
\begin{proof}
Using a partition tree, Matou\v{s}ek~\cite{ref:MatousekEf92} built a data structure to answer halfplane range searching queries among a set of points in the plane. In the dual problem, we are given a set of lines in the plane, and each query asks for the number of lines below a query point. Matou\v{s}ek~\cite{ref:MatousekEf92} built a data structure of $O(n)$ space in $O(n\log n)$ time, with $O(n^{1/2}\log^{O(1)} n)$ query time. If we wish to report all these planes below the query point, then the query time becomes $O(n^{1/2}\log^{O(1)} n + k)$, where $k$ is the output size. As discussed in~\cite[Theorem~7.1]{ref:MatousekEf92}, the data structure can be easily extended to accommodate line deletions (and insertions) and each deletion can be handled in $O(\log n)$ amortized time.

Wang~\cite{ref:WangUn23} generalized the partition tree technique of Matou\v{s}ek~\cite{ref:MatousekEf92} to queries among a set of arcs like those in $\calA$ with asymptotically the same complexities. We can then use the same method as discussed in \cite[Theorem~7.1]{ref:MatousekEf92} to handle deletions in $O(\log n)$ amortized time each. 
\end{proof}


\begin{lemma}\label{lem:delonly}
There is a data structure of $O(n)$ size to maintain a set $\calA$ of $n$ arcs to support $O(\log n)$ amortized time deletions and $O(\sqrt{n}\log^{O(1)}n  +k)$ time $k$-lowest-arcs queries. If a set $\Gamma$ of $n$ arcs is given initially, the data structure can be constructed in $O(n\log n)$ time. 
\end{lemma}
\begin{proof}
We first show that the data structure can be computed in $O(n\log n)$ time and $O(n)$ space for a set $\calA$ of $n$ arcs, and then discuss the deletions. 

Let $k_i=2^i$ for $i=0,1,\ldots,\lceil\log n\rceil$. We compute $k_i$-shallow $(Ck_i/n)$-cuttings $\Xi_{i}$ for $\calA$, for all $i=0,1,\ldots,\lceil\log n\rceil$, for a constant $C$. This takes $O(n\log n)$ time by Theorem~\ref{theo:shallowcut}. We do not store the conflict lists for the cuttings. Therefore, the total space is $O(n)$. Instead, we implicitly maintain the conflict lists of all cuttings by building the data structure of Lemma~\ref{lem:delonlypoint} for $\calA$, denoted by $\calD$, in $O(n)$ space and $O(n\log n)$ time. 

Given a query vertical line $\ell^*$ and a number $k$, we wish to compute the set $\calA_{\ell^*}$ of $k$ lowest arcs of $\calA$ intersecting $\ell^*$. To this end, we first find the smallest $i$ such that $k\leq k_i$, which takes $O(\log n)$ time. By definition, $k_i=O(k)$. Then, in $O(\log n)$ time we can find the cell $\Delta_{\ell^*}$ of $\Xi_i$ intersecting $\ell^*$. Let $q$ be the intersection between the top edge of $\Delta_{\ell^*}$ and $\ell^*$. By the definition of shallow cutting $\Xi_i$, $q$ is above at least $k_i$ arcs of $\calA$. As $k\leq k_i$, $\calA_{\ell^*}$ is a subset of the set $\calA_q$ of arcs of $\calA$ below $q$. Using $\calD$, we compute $\calA_q$ in $O(\sqrt{n}\log^{O(1)}n + |\calA_q|)$ time by Lemma~\ref{lem:delonlypoint}. By the definition of $\Xi_i$, $|\calA_q|=O(k_i)$, and thus $|\calA_q|=O(k)$ as $k_i=O(k)$. Finally, we find the $k$ lowest arcs at $\ell^*$ among the arcs of $\calA_q$, which can be done in $O(|\calA_q|)$ time using the linear time selection algorithm. Since $|\calA_q|=O(k)$, the total query time is $O(\sqrt{n}\log^{O(1)}n+k)$. 

To delete an arc $\gamma$, we first delete it from $\calD$, which takes $O(\log n)$ amortized time by Lemma~\ref{lem:delonlypoint}. After $n/2$ deletions, we reconstruct the entire data structure from scratch. As such, the amortized deletion time is $O(\log n)$. The lemma thus follows. 
\end{proof}

We have the following lemma for another deletion-only data structure, obtained by following the same algorithmic scheme as \cite[Lemma~6]{ref:BergDy23} and replacing their shallow cutting algorithm for lines with our shallow cutting algorithm for arcs of $\calA$ in Theorem~\ref{theo:shallowcut}. 

\begin{lemma}\label{lem:delonly20}{\em\cite[Lemma~6]{ref:BergDy23}}
For any fixed $r$, there is a data structure of $O(n\log r)$ size to maintain a set $\calA$ of $n$ arcs to support $O(r\log n)$ amortized time deletions and $O(\log r + n/r  +k)$ time $k$-lowest-arcs queries. If a set $\Gamma$ of $n$ arcs is given initially, the data structure can be constructed in $O(n\log n)$ time. 
\end{lemma}

\subsubsection{Fully-dynamic data structure for Lemma~\ref{lem:karcquery}}

With the two deletion-only data structures in Lemmas~\ref{lem:delonly} and \ref{lem:delonly20}, we are now in a position to describe our fully dynamic data structure for Lemma~\ref{lem:karcquery}. 

\paragraph{Overview.}
To achieve our result in Lemma~\ref{lem:karcquery}, roughly speaking, we can simply plug our shallow cutting algorithm for $\calA$ in Theorem~\ref{theo:shallowcut} and Lemma~\ref{lem:delonly} into the algorithmic scheme of \cite{ref:BergDy23} or \cite{ref:ChanTh12}. The algorithms of \cite{ref:BergDy23} and \cite{ref:ChanTh12} are similar. For the method of \cite{ref:ChanTh12}, we can just replace their shallow cutting algorithm for lines~\cite{ref:ChanOp16} with our shallow cutting algorithm for $\calA$ in Theorem~\ref{theo:shallowcut} and replace their deletion-only data structure~\cite{ref:MatousekEf92} with a combination of Lemmas~\ref{lem:delonly} and \ref{lem:delonly20}.
In addition, a general technique of querying multiple structures simultaneously from \cite[Theorem~1]{ref:BergDy23} is also needed. For the method of \cite{ref:BergDy23}, it was described for the plane problem in 3D with a query time $O(\log^2 n/\log\log n+k)$. We can follow the same algorithmic scheme but using our shallow cutting algorithm for $\calA$ in Theorem~\ref{theo:shallowcut} and Lemma~\ref{lem:delonly} in the corresponding places. In addition, since our problem is a 2D problem, the technique of dynamic interval trees utilized in \cite{ref:ChanOp16} can be used to reduce the query time component from $O(\log^2 n/\log\log n)$ to $O(\log n)$. In the following, we adapt the method from Chan~\cite{ref:ChanTh12}. 
\medskip

The data structure is an adaptation of the one originally for dynamic 3D convex hulls \cite{ref:ChanA10,ref:ChanDy20,ref:KaplanDy20} (the original idea was given in \cite{ref:ChanA10} and subsequent improvements were made in \cite{ref:ChanDy20,ref:KaplanDy20}). We first give the following lemma. The lemma is similar to \cite[Theorem~3.1]{ref:ChanTh12}, which is based on the result in \cite{ref:ChanA10}, but Lemma~\ref{lem:cuttingproperty} provides slightly better complexities than \cite[Theorem~3.1]{ref:ChanTh12} by using the recently improved result of \cite{ref:ChanDy20}. 

\begin{lemma}\label{lem:cuttingproperty}\emph{\cite{ref:ChanA10,ref:ChanTh12,ref:ChanDy20,ref:KaplanDy20}}
Let $\calA$ be a set of arcs, which initially is $\emptyset$ and undergoes $n$ updates (insertions and deletions). For any $b\geq 2$, we can maintain a collection of shallow cuttings $T_i^{j}$ in the bottom-open pseudo-trapezoid form, $i=1,2,\ldots,\lceil\log n\rceil$, $j=1,2,\ldots,O(\log_bn)$, in $b^{O(1)}\log^5 n$ amortized time per update such that the following properties hold. 
\begin{enumerate}
\item Each cutting $T_i^j$ is of size $O(2^i)$ and 
never changes until it is replaced by a new one created from scratch. The total size of all cuttings created over time is $b^{O(1)}\log^3 n$. 
\item Each cell $\Delta\in T_i^j$ is associated with a list $L_{\Delta}$ of $O(n/2^i)$ arcs of $\calA$. Each list $L_{\Delta}$ undergoes deletions only after its creation. The total size of all such lists created over time is $b^{O(1)}n\log^4 n$. 
\item For any $k\geq 1$, let $i_k=\lceil\log(n/Ck)\rceil$ for a sufficiently large constant $C$. For any vertical line $\ell^*$, if an arc $\gamma\in \calA$ is among the $k$ lowest arcs at $\ell^*$, then there exists some $j$ such that $\gamma$ is in the list $L_{\Delta^j}$ of the cell $\Delta^j\in T_{i_k}^j$ intersecting $\ell^*$. 
\item 
At any moment, for each $i$, the number of cells of the current cuttings $T_i^j$ for all $j$ is $O(2^i)$. This implies that the total size of the lists $L_{\Delta}$ of all cells $\Delta$ of all current cuttings $T_i^j$ at any moment is $O(n\log n)$. 
\end{enumerate}
\end{lemma}
\begin{proof}
The lemma follows \cite[Theorem~3.1]{ref:ChanTh12}, which is based on the result in \cite{ref:ChanA10}. Following an observation in \cite{ref:KaplanDy20} along with other observations, the result in \cite{ref:ChanA10} is improved in \cite{ref:ChanDy20}. The complexities of the lemma are obtained by following the same algorithmic scheme of \cite[Theorem~3.1]{ref:ChanTh12} with the improved strategy of \cite{ref:ChanDy20}, which relies on an input-restricted shallow cutting algorithm for lines/planes based on a slight modification of the algorithm in \cite{ref:ChanOp16}. Here, we replace their shallow cutting algorithm by ours in Theorem~\ref{theo:shallowcut} for $\calA$. Everything else is the same. If we follow exactly the time complexities in \cite{ref:ChanDy20}, then the amortized time per update would be $b^{O(1)}\log^4n$. However, here the $b^{O(1)}\log^5n$ bound suffices for our purpose as other parts of the algorithm dominate the overall time complexity. As such, we can afford to compute a single shallow cutting in $O(n\log n)$ time by Theorem~\ref{theo:shallowcut} instead of resorting to the input-restricted shallow cutting algorithm as in \cite{ref:ChanDy20}.
\end{proof}

With Lemma~\ref{lem:cuttingproperty}, we can answer a $k$-lowest-arcs query as follows. Consider a query vertical line $\ell^*$. By Lemma~\ref{lem:cuttingproperty}(3),  for each $j$, we compute the cell $\Delta^j$ of $T_{i_k}^j$ intersecting $\ell^*$, which takes $O(\log n)$ time by binary search as the $x$-projections of $T_{i_k}^j$ partition the $x$-axis into intervals. Then, we use ``brute-force'' to find the $k$ lowest arcs among all arcs in $L_{\Delta^j}$ in $O(k)$ time as $|L_{\Delta^j}|=O(k)$. Finally, among all arcs found above, we return the $k$ lowest arcs, which takes $O(k\log_bn)$ time. As such, the total query time is $O((\log n+k)\log_bn)$. 

\paragraph{An improved query algorithm.}
We now improve the query time. We store each list $L_{\Delta}$ by a deletion-only data structure that supports $k$-lowest-arcs queries. Suppose that such a deletion-only data structure is of space $S_0(|L_{\Delta}|)$, supports each $k$-lowest-arcs query in $O(Q_0(|L_{\Delta}|)+k)$ time and $D_0(|L_{\Delta}|)$ deletion time, and can be built in $O(P_0(|L_{\Delta}|))$ time. Then, by Lemma~\ref{lem:cuttingproperty}(2), each update causes at most $b^{O(1)}\log^4 n$ amortized number of deletions to the lists $L_{\Delta}$, and thus the amortized update time is 
\begin{equation}\label{equ:update}
    U(n) = b^{O(1)} \log^5 n + \max_{\Delta\in T_i^j} D_0(|L_{\Delta}|)\cdot b^{O(1)}\log^4 n + P_0(b^{O(1)}n\log^4 n)/n.
\end{equation}
Note that the last term is obtained due to the following. After every $n$ updates, we reconstruct the entire data structure and thus the reconstruction time is on the order of $\sum_{\Delta\in T_i^j}P_0(|L_{\Delta}|)$, which is bounded by $P_0(b^{O(1)}n\log^4 n)$ since $\sum_{\Delta\in T_i^j}|L_{\Delta}|=b^{O(1)}n\log^4 n$ by Lemma~\ref{lem:cuttingproperty}(2) (assuming that $P_0(n)=\Omega(n)$).

By Lemma~\ref{lem:cuttingproperty}(4), the total space is 
\begin{equation}\label{equ:space}
    S(n) = O\left(\sum_{i=1}^{\log n}2^i\cdot S_0(n/2^i)\right).
\end{equation}

For each query, there are two tasks: (1) Compute the cell $\Delta^j$ for every $j$; (2) find the $k$ lowest arcs of all lists $L_{\Delta^j}$ for all $j$. In the following, we solve the first task in $O(\log n+\log_b n)$ time and solve the second task in $O(\log n + k)$ time. 

For the first task, following the method in \cite{ref:ChanTh12}, for each $i$, we use a dynamic interval tree~\cite{ref:deBergCo08} to store the intervals of the $x$-projections of the cuttings $T_i^j$ for all $j$ in an interval tree $T_i$. Using $T_i$, all $t$ intervals intersecting $\ell^*$ can be computed in $O(\log n + t)$ time. In our problem, $t=O(\log_b n)$. Insertions and deletions of intervals on $T_i$ can be supported in $O(\log n)$ amortized time. We can thus maintain all interval trees $T_i$ in additional amortized $\log n\cdot b^{O(1)}\log^3 n$ time per update as the total size of all cuttings over time is $b^{O(1)}\log^3 n$ by Lemma~\ref{lem:cuttingproperty}(1). This additional update time is subsumed by the first term in \eqref{equ:update}. In this way, the first task can be solved in $O(\log n+\log_b n)$ time. 

For the second task, instead of brute-force, we use the deletion-only data structures for the lists $L_{\Delta^j}$ and resort to a technique of querying multiple $k$-lowest-arcs data structures simultaneously in \cite[Theorem~1]{ref:BergDy23}, which is based on an adaption of the heap selection algorithm of Frederickson~\cite{ref:FredericksonAn93}. Applying \cite[Theorem~1]{ref:BergDy23}, the second task can be accomplished in $O(k+\log_bn \cdot \max_{j}Q_0(|L_{\Delta^j}|))$ time, which is $O(k+\log_bn \cdot Q_0(O(k)))$ since the size of each $|L_{\Delta^j}|$ is $O(k)$. 

Combining the complexities of the first and second tasks, the overall query time is 
\begin{equation}\label{equ:query}
Q(n)=O(\log n+\log_b n+ k) + Q_0(O(k))\cdot \log_bn.
\end{equation}

Let $m=|L_{\Delta}|$. Depending on the value of $m$, we use different deletion-only data structures for $L_{\Delta}$. 
\begin{enumerate}
\item If $m\geq \log^3 n$, then we use Lemma~\ref{lem:delonly} to handle $L_{\Delta}$ with $P_0(m)=O(m\log m)$, $S_0(m)=O(m)$, $D_0(m)=O(\log m)$, $Q_0(m)=m^{1/2}\log^{O(1)} m$. Plugging them into \eqref{equ:update}, \eqref{equ:space}, and \eqref{equ:query} and setting $b=\log^{\epsilon}n$, we obtain $U(n)=O(\log^{5+\epsilon}n)$, $S(n)=O(n\log n)$, and $Q(n)=O(\log n+k+k^{1/2}\log^{O(1)}k\cdot \log n/\log\log n)$. Since $m=O(k)$, we have $k=\Omega(m)$. As $m\geq \log^3 n$, we have $k=\Omega(\log^3 n)$. Therefore, $Q(n)=O(\log n+k)$.

\item 
If $m<\log^3 n$, then we use Lemma~\ref{lem:delonly20} to handle $L_{\Delta}$ by setting $r=\log n/\log\log n$. This results in $P_0(m)=O(m\log m)$, $S_0(m)=O(m\log\log n)$, $D_0(m)=O(\log n\log m/\log\log n)$, $Q_0(m)=O(\log\log n+m\log\log n/\log n)$. Since $m<\log^3 m$, we have $D_0(m)=O(\log n)$. 
Plugging them into \eqref{equ:update} and \eqref{equ:query} and setting $b=\log^{\epsilon}n$, we obtain $U(n)=O(\log^{5+\epsilon}n)$ and 
$Q(n)=O(\log n+k+(\log\log n+k\log\log n/\log n)\cdot \log n/\log\log n)$, which is $O(\log n + k)$. For the space, since $m<\log^3n$, if we plug $S_0(m)=O(m\log\log n)$ into \eqref{equ:space}, we only need to consider those $i$'s such that $n/2^i < \log^3 n$. There are only $O(\log\log n)$ such $i$'s. Therefore, we obtain $S(n)=O(n\log^2\log n)$ for all such $m$'s in this case. 
\end{enumerate}

Combining the above two cases leads to $U(n)=O(\log^{5+\epsilon}n)$, $S(n)=O(n\log n)$, and $Q(n)=O(\log n + k)$. We can actually obtain a better bound for the insertion time. If $P'(n)$ is the preprocessing time for constructing the data structure for a set of $n$ arcs, then the amortized insertion time $I(n)$ is bounded by $I(n)=O(b\log_b n\cdot P'(n)/n)$~\cite{ref:BergDy23,ref:ChanDy20}. According to our above discussion and Lemma~\ref{lem:cuttingproperty}(4), constructing the deletion-only data structures for all lists $L_{\Delta}$ is $O(n\log^2 n)$. The shallow cuttings in Lemma~\ref{lem:cuttingproperty} can be built in $O(n\log^2 n)$ following the method in \cite{ref:ChanDy20} and using our shallow cutting algorithm in Theorem~\ref{theo:shallowcut}. In this way, we can bound the amortized insertion time by $O(b\log_b n\log^2 n)$, which is $O(\log^{3+\epsilon}n)$ with $b=\log^{\epsilon}n$. This proves Lemma~\ref{lem:karcquery}.

\section{Algorithm for shallow cuttings}
\label{sec:algoshallow}
In this section, we prove Theorem~\ref{theo:shallowcut}. We follow the notation in Section~\ref{sec:defshallow}.

As in \cite{ref:ChanOp16}, we use parameter $K=n/r$ instead of $r$. A $k$-shallow $(1/r)$-cutting becomes a $k$-shallow $(K/n)$-cutting and we use a $(k,K)$-shallow cutting to represent it. It has the following properties: (1) $|\calA_{\Delta}|\leq K$ for each cell $\Delta$; (2) the union of all cells covers $L_{\leq k}(\calA)$. Theorem~\ref{theo:shallowcut} says that a $(k,O(k))$-shallow cutting of size $O(n/k)$ in the bottom-open pseudo-trapezoid form can be computed in $O(n\log\frac{n}{k})$ time. 

Following the analysis of Matou\v{s}ek~\cite{ref:MatousekRe92}, we start with the following lemma, which shows the existence of the shallow cutting in {\em the pseudo-trapezoid form}, i.e., each cell is a pseudo-trapezoid but not necessarily bottom-open. 
\begin{lemma}\label{lem:shallowexist10}
For any $k\in [1,n]$, there exists a $(k,O(k))$-shallow cutting of size $O(n/k)$ in the pseudo-trapezoid form. 
\end{lemma}
\begin{proof}
We can basically follow the proof of \cite[Theorem~2.1]{ref:MatousekRe92}, which uses the random sampling techniques of Chazelle and Friedman~\cite{ref:ChazelleA90}. The proof of \cite{ref:MatousekRe92} is originally for lines in 2D or for hyperplanes in the high-dimensional space. We can generalizes the 2D analysis to arcs of $\calA$. A {\em pseudo-trapezoidal decomposition} for a subset $\calA'\subseteq \calA$ of arcs is to draw a vertical line through each vertex of the arrangement of $\calA'$ until the line hits an arc above the vertex and also hits an arc below the vertex. Each cell of the decomposition is thus a pseudo-trapezoid. We use pseudo-trapezoidal decomposition to replace canonical triangulation for lines used in \cite{ref:MatousekRe92} and then similar properties to those in \cite[Lemma~2.3]{ref:MatousekRe92} still hold. Then, we can follow similar analysis to obtain \cite[Lemma~2.4]{ref:MatousekRe92}, which relies on \cite[Lemma~2.3]{ref:MatousekRe92}. Consequently, \cite[Theorem~2.1]{ref:MatousekRe92} can be proved by using \cite[Lemma~2.4]{ref:MatousekRe92}. In the proof, instead of using canonical triangulation for lines, we use pseudo-trapezoidal decomposition for arcs of $\calA$. The proof also relies on the existence of a standard $(1/t)$-cutting of size $O(t^2)$ in the pseudo-trapezoid form for a subset of arcs of $\calA$. It is already known that such a cutting exists~\cite{ref:AgarwalPs05,ref:WangUn23}. In addition, the proof needs an $O(nk)$ bound on the number of vertices of $L_{\leq k}(\calA)$. Such a bound holds according to the result of Sharir~\cite{ref:SharirOn91}, since each arc of $\calA$ is $x$-monotone and every two arcs intersect at most once. 
Therefore, following the analysis of \cite{ref:MatousekRe92}, the lemma can be proved. 
\end{proof}

Chan and Tsakalidis~\cite{ref:ChanOp16} introduced a vertex form of the shallow cutting for lines. Here for arcs of $\calA$, using vertices is not sufficient. We will introduce a vertex-segment form in Section~\ref{sec:versegform}. The definition requires a concept, which we call {\em line-separated $\alpha$-hulls} and is discussed in Section~\ref{sec:alphahull}. In Section~\ref{subsec:algoshallowcut}, we present our algorithm to compute shallow cuttings in the vertex-segment form. 

\subsection{Line-separated $\alpha$-hulls}
\label{sec:alphahull}
The line-separated $\alpha$-hull is an extension of the $\alpha$-hull introduced in \cite{ref:EdelsbrunnerOn83}. 
In \cite{ref:EdelsbrunnerOn83}, $\alpha$-hull is considered for all values $\alpha\in (-\infty,\infty)$. For our problem, we only consider the value $\alpha=-1$.

Let $Q$ be a set of points in $\bbR^-$. We define the {\em line-separated $\alpha$-hull} $H_{\ell}(Q)$ of $Q$ with respect to the $x$-axis $\ell$ as the complement of the union of all unit disks with centers in $\bbR^+$ that do not contain any point of $Q$ (so the disk centers and the points of $Q$ are separated by $\ell$, which is why we use ``line-separated''; see Fig.~\ref{fig:linehull}).

\begin{figure}[t]
\begin{minipage}[t]{\linewidth}
\begin{center}
\includegraphics[totalheight=1.0in]{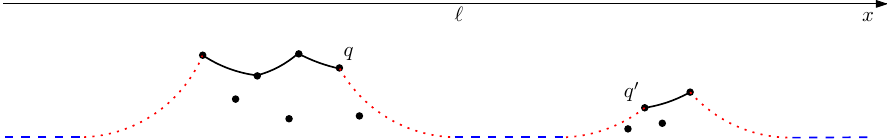}
\caption{Illustration the boundary of $H_{\ell}(Q)$, where $Q$ is the set of points below the $x$-axis $\ell$. It consists of three (blue) dashed horizontal line segments of $y$-coordinates $-1$, four (red) dotted $\bbR^+$-constrained arcs with centers on  $\ell$, and four other solid $\bbR^+$-constrained arcs. The region below the boundary is $H_{\ell}(Q)$.}
\label{fig:linehull}
\end{center}
\end{minipage}
\end{figure}

Many of the properties of the $\alpha$-hulls~\cite{ref:EdelsbrunnerOn83} can be extended to the line-separated case. We list some of these in the following observation; the proof is a straightforward extension of that in \cite{ref:EdelsbrunnerOn83} by adding the ``line-separated'' constraint. 

\begin{observation}\label{obser:hull}
 \begin{enumerate}
    \item $Q\subseteq H_{\ell}(Q)$, and for any subset $Q'\subseteq Q$, $H_{\ell}(Q')\subseteq H_{\ell}(Q)$. 
     \item A point $q\in Q$ is a vertex of $H_{\ell}(Q)$ if and only if there exists a unit disk with center in $\bbR^+$ and its boundary containing $q$ such that the interior of the disk does not contain any point of $Q$. 
     \item If there is an $\bbR^+$-constrained arc connecting two points of $Q$ such that the interior of the underlying disk of the arc does not contain any point of $Q$, then the arc is an edge of $H_{\ell}(Q)$.     
 \end{enumerate}   
\end{observation}

For any two points $q,q'\in \bbR^-$ that can be covered by a $\bbR^+$-constrained unit disk, there exists a unique $\bbR^+$-constrained arc that connects $q$ and $q'$; we use $\gamma(q,q')$ to denote that arc.



\subsubsection{Algorithm for computing $H_{\ell}(Q)$}
By slightly modifying the algorithm of \cite{ref:EdelsbrunnerOn83}, $H_{\ell}(Q)$ can be computed in $O(m\log m)$ time, where $m=|Q|$. The algorithm also suggests that $\partial H_{\ell}(Q)$ is $x$-monotone. 
In the following, assuming that the points of $Q$ are already sorted from left to right as $q_1,q_2,\ldots,q_m$, we give a linear time algorithm to compute $H_{\ell}(Q)$, which is similar in spirit to Graham's scan for computing convex hulls. 

\paragraph{Irrelevant points.}
Note that if a point $q\in Q$ whose $y$-coordinate is smaller than or equal to $-1$, then $q$ must be in $H_{\ell}(Q)$ because every $\bbR^+$-constrained disk does not contain $q$ in the interior. Hence, in that case $q$ is {\em irrelevant} for computing $H_{\ell}(Q)$ and thus can be ignored. If all points of $Q$ are irrelevant, then $H_{\ell}(Q)$ is simply the region below the horizontal line whose $y$-coordinate is $-1$. In the following, we assume that every point of $Q$ is relevant. 


\paragraph{Wings.}
Consider a point $q\in Q$. Let $a$ be a point on the $x$-axis $\ell$ with $x(a)<x(q)$ such that $q$ is on the unit circle $C_a$ centered at $a$. Let $p$ be the lowest point of $C_a$. We define the {\em left wing} of $q$ to be the concatenation of the following two parts (see Fig.~\ref{fig:wing}): (1) the arc of $C_a \cap R^-$ between $q$ and $p$, called the {\em left wing arc}, and the horizontal half-line with $p$ as the right endpoint, called the {\em left wing half-line}. The point $p$ is called the {\em left wing vertex} of $q$. We define the {\em right wing} of $q$ and the corresponding concepts symmetrically. The left and right wings together actually form the boundary of the line-separated $\alpha$-hull of $\{q\}$.  In Fig.~\ref{fig:linehull}, the four (red) dotted arcs are wing arcs and the three (blue) dashed segments are on wing half-lines. 

\begin{figure}[t]
\begin{minipage}[t]{\linewidth}
\begin{center}
\includegraphics[totalheight=1.0in]{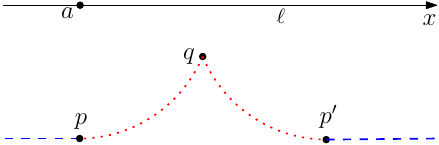}
\caption{Illustration the wings of the a point $q$. The two (red) dotted curves are wing arcs and the two (blue) dashed segments are wing half-lines. $p$ and $p'$ are the left and right wing vertices, respectively.}
\label{fig:wing}
\end{center}
\end{minipage}
\end{figure}

\paragraph{Far-away position.}
Consider two points $q,q'\in \bbR^-$ such that $x(q)<x(q')$. We say that $(q,q')$ are in {\em far-away} position if $x(p)<x(p')$ holds, where $p$ is the right wing vertex of $q$ and $p'$ is the left wing vertex of $q'$ (see Fig.~\ref{fig:farawary}). 
In this case, $q'$ is above the right wing of $q$ and there is no $\bbR^+$-constrained unit disk covering both $q$ and $q'$. 
We use $\beta(q,q')$ to denote the concatenation of the right wing arc of $q$, the segment $\overline{pp'}$, and the left wing arc of $q'$. In fact, the left wing of $q$, $\beta(q,q')$, and the right wing of $q'$ together form the boundary of the line-separated $\alpha$-hull of $\{q,q'\}$. In Fig.~\ref{fig:linehull}, the two points $q$ and $q'$ are also in far-away position. 

\begin{figure}[t]
\begin{minipage}[t]{\linewidth}
\begin{center}
\includegraphics[totalheight=1.0in]{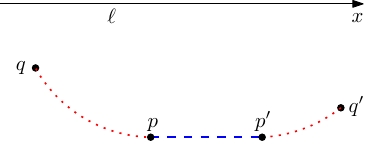}
\caption{Illustration two points $q$ and $q'$ that are in far-away position. The two (red) dotted arcs and the (blue) dashed segments in between constitute $\beta(q,q')$.}
\label{fig:farawary}
\end{center}
\end{minipage}
\end{figure}

\paragraph{The algorithm.}
Define $Q_i=\{q_1,\ldots,q_i\}$ for each $1\leq i\leq m$. Our algorithm handles the points of $Q$ incrementally from $q_1$ to $q_m$. For each $q_i$, the algorithm computes $H_{\ell}(Q_i)$ by updating $H_{\ell}(Q_{i-1})$ with $q_i$.  
Suppose that $q_{i_1},q_{i_2},\ldots,q_{i_t}$ are the points of $Q_i$ that are the vertices of $H_{\ell}(Q_{i})$ sorted from left to right. Then, our algorithm maintains the following invariant: the boundary $\partial H_{\ell}(Q_{i})$ is $x$-monotone and consists of the following parts from left to right: the left wing of $q_{i_1}$, $\gamma(q_{i_j},q_{i_{j+1}})$ or $\beta(q_{i_j},q_{i_{j+1}})$, for each $j=1,2,\ldots,t-1$ in order, and the right wing of $q_{i_t}$. 
$H_{\ell}(Q_i)$ is the region below $\partial H_{\ell}(Q_{i})$.


Initially, for $q_1$, we set $\partial H_{\ell}(Q_{1})$ to the concatenation of the left wing and the right wing of $q_1$. 
In general, suppose we already have $\partial H_{\ell}(Q_{i-1})$. We compute $\partial H_{\ell}(Q_{i})$ as follows. We process the points of $q_{i_1},q_{i_2},\ldots,q_{i_t}$ in the backward order. For ease of exposition, we assume that $t>1$; the special case $t=1$ can be easily handled.

We first process the point $q_{i_t}$.  
If $q_{i_t}$ and $q_i$ are in the far-away position, then we delete the right wing of $q_{i_t}$ from $H_{\ell}(Q_{i-1})$ and add $\beta(q_{i_t},q_i)$ and the right wing of $q_i$. This finishes computing $H_{\ell}(Q_i)$. Below, we assume that $q_{i_t}$ and $q_i$ are not in the far-away position. 

If $q_i$ is below the right wing of $q_{i_t}$, then $q_i$ is inside $H_{\ell}(Q_{i-1})$. In this case, $H_{\ell}(Q_{i})$ is $H_{\ell}(Q_{i-1})$ and we are done. If $q_i$ is above the right wing of $q_{i_t}$, then we further check whether the arc $\gamma(q_{i_t},q_i)$ exists (which is true if and only if there exists an $\bbR^+$-constrained unit disk covering both $q_{i_t}$ and $q_i$). 

\begin{itemize}
    \item If $\gamma(q_{i_t},q_i)$ does not exist (in this case $q_{i_t}$ must be below the left wing of $q_i$ and thus $q_{i_t}$ does not contribute to $H_{\ell}(Q_i)$ because it is ``dominated'' by $q_i$), then we ``prune'' $q_{i_t}$ from $\partial H_{\ell}(Q_{i-1})$, i.e., delete the right wing of $q_{i_t}$ and also delete $\gamma(q_{i_{t-1}},q_{i_t})$ or $\beta(q_{i_{t-1}},q_{i_t})$ whichever exists in $H_{\ell}(Q_{i-1})$. Next, we process $q_{i_{t-1}}$ following the same algorithm. 
    \item If $\gamma(q_{i_t},q_i)$ exists, 
    then we further check whether $D$ contains $q_{i_{t-1}}$, where $D$ is the underlying disk of $\gamma(q_{i_t},q_i)$. If $q_{i_{t-1}}\in D$, then $q_{i_{t}}$ must be in $H_{\ell}(\{q_{i_{t-1}},q_i\})$ and thus does not contribute to $H_{\ell}(Q_i)$. In this case, we prune $q_{i_t}$ as above and continue processing $q_{i_{t-1}}$. If $q_{i_{t-1}}\not\in D$, we delete the right wing of $q_{i_t}$ from $\partial H_{\ell}(Q_{i-1})$ and add the arc $\gamma(q_{i_t},q_i)$ and the right wing of $q_i$; this finishes computing $H_{\ell}(Q_i)$. 
\end{itemize}

Clearly, the runtime for computing $H_{\ell}(Q_i)$ is $O(1+t')$, where $t'$ is the number of points of $q_{i_1},q_{i_2},\ldots,q_{i_t}$ pruned from $H_{\ell}(Q_{i-1})$. 
The overall algorithm for computing $H_{\ell}(Q)$ takes $O(m)$ time since once a point is pruned it will never appear on the hull again, which resembles Graham's scan for computing convex hulls. 

\subsubsection{Vertical decompositions}

According to the above discussion, $H_{\ell}(Q)$ has at most $5m$ vertices with $m=|Q|$, including all wing vertices. The vertical downward rays from all vertices partition $H_{\ell}(Q)$ into at most $5m$ bottom-open pseudo-trapezoids and rectangles. We call this partition the {\em vertical decomposition} of $H_{\ell}(Q)$, denoted by $\vd(Q)$. 

In our later discussion, we need to combine $Q$ with a set $S$ of pairwise disjoint segments on $\ell$ whose endpoints are all in $Q$. For each segment $s\in S$, we draw a vertical downward ray from each endpoint of $s$; let $R(s)$ denote the bottom-open rectangular region bounded by the two rays and $s$. The regions $R(s)$ for all segments $s\in S$ form the {\em vertical decomposition} of $S$, denoted by $\vd(S)$. 

We combine  $\vd(Q)$ and $\vd(S)$ to form a vertical decomposition of $Q$ and $S$, denoted by $\vd(Q\cup S)$ as follows. Let $U$ be the upper envelope of $H_{\ell}(Q)$ and $S$. We draw a vertical downward ray from each vertex of $v$ of $U$. These rays divide the region below $U$ into cells, each of which is bounded by two vertical rays from left and right, and bounded from above by a line segment or an $\bbR^+$-constrained arc. 
These cells together form the vertical decomposition $\vd(Q\cup S)$. 

In the following, depending on the context, $\vd(Q)$ may refer to the region covered by all cells of it; the same applies to $\vd(S)$ and $\vd(Q\cup S)$. As such, we have $\vd(Q\cup S)=\vd(Q)\cup \vd(S)$. Note that since the endpoints of all segments of $S$ are in $Q$ and on $\ell$, the boundary $\partial \vd(Q\cup S)$ is $x$-monotone. 

\subsection{Shallow cuttings in the vertex-segment form}
\label{sec:versegform}
We introduce a vertex-segment form of the shallow cutting. Given parameters $k,K\in [1,n]$ with $k\leq K$, a $(k,K)$-shallow cutting for the arcs of $\calA$ in the {\em vertex-segment} form is a set $Q$ of points in $\bbR^-$ along with a set $S$ of interior pairwise-disjoint segments on $\ell$ such that the following conditions hold:
\begin{enumerate}
    \item The endpoints of all segments of $S$ are in $Q$. 
    \item Every point of $Q$ has level at most $K$ in $\calA$.
    \item Every segment of $S$ intersects at most $K$ arcs of $\calA$.  
    \item $\vd(Q\cup S)$ covers $L_{\leq k}(\calA)$.
\end{enumerate}
The {\em conflict list} of a point $q \in Q$, denoted by $\calA_q$, is the set of arcs of $\calA$ below $q$. Note that $|\calA_q|\leq K$ as the level of $q$ is at most $K$. The {\em conflict list} of a segment $s\in S$, denoted by $\calA_s$, is the set of arcs intersecting $s$. The conflict lists of $(Q,S)$ refer to the conflict lists of all points of $Q$ and all segments of $S$. The {\em size} of the cutting is defined to be $|Q|$. Observe that since the endpoints of all segments of $S$ are in $Q$ and the segments of $S$ are interior pairwise-disjoint, we have $|S|<|Q|$. Therefore, $\vd(Q\cup S)$ has $O(|Q|)$ cells, and more specifically, at most $5|Q|$ cells. Further, we have following observation. 

\begin{observation}\label{obser:cellconlist}
Suppose that $(Q,S)$ is a $(k,K)$-shallow cutting for $\calA$ in the vertex-segment form. Then every cell of $\vd(Q\cup S)$ intersects at most $3K$ arcs of $\calA$. 
\end{observation}
\begin{proof}
Consider a cell $\Delta$ of $\vd(Q\cup S)$. 
Let $e$ be the top edge of $\Delta$. By the definition of $\vd(Q\cup S)$, $e$ is one of the following: a segment of $S$, an arc $\gamma(q,q')$ for two points $q,q'\in Q$, a wing arc of a point $q\in Q$, and a segment of a wing half-line of a point $q\in Q$. 
Below, we argue $|\calA_{\Delta}|\leq 3K$ for each of these cases, where $\calA_{\Delta}$ is the set of arcs of $\calA$ intersecting $\Delta$. 

\begin{enumerate}
    \item If $e$ a segment of $S$, then recall that $|\calA_e|\leq K$. Also, the size of the conflict list of each endpoint of $e$ is at most $K$. For any arc $\gamma\in \calA$ intersecting $\Delta$, since the center of $\gamma$ is in $\bbR^+$, $\gamma$ must either intersect $e$ or in the conflict list of at least one endpoint of $e$. Therefore, $|\calA_{\Delta}|\leq 3K$ holds. 
    
    \item If $e$ is an arc $\gamma(q,q')$ for two points $q,q'\in Q$, then any arc $\gamma\in \calA$ intersecting $\Delta$ must be in the conflict list of one of $q$ and $q'$. Hence, we have $|\calA_{\Delta}|\leq 2K$.

    \item If $e$ is a wing arc $\gamma$ of a point $q\in Q$, then $q$ is an endpoint of $\gamma$. Let $p$ be the other endpoint of $\gamma$. By definition, the $y$-coordinate of $p$ is $-1$. Thus, no arc of $\calA$ is below $p$. By definition, the radius of $\gamma$ is $1$ and the center of $\gamma$ is on $\ell$. Hence, any arc of $\calA$ intersecting $\Delta$ must be in the conflict list of $q$ and thus $|\calA_{\Delta}|\leq |\calA_q|\leq K$.

    \item If $e$ is a segment $s$ of a wing half-line of a point $q\in Q$, then by definition $e$ is horizontal and has $y$-coordinate equal to $-1$. As centers of all arcs of $\calA$ are in $\bbR^+$, no arc of $\calA$ can intersect $\Delta$. Hence, $|\calA_{\Delta}|= 0$.   
\end{enumerate}
Combining all the above cases leads to $|\calA_{\Delta}|\leq 3K$. 
\end{proof}

In the next two lemmas, we show that shallow cuttings in the vertex-segment form and in the pseudo-trapezoid form can be transformed to each other. 

\begin{lemma}\label{lem:regulartovertex}
A $(k,K)$-shallow cutting of size $t$ in the pseudo-trapezoid form can be transformed into a $(k,k+K)$-shallow cutting in the vertex-segment form of size $O(t)$.     
\end{lemma}
\begin{proof}
Let $\Xi$ be a $(k,K)$-shallow cutting of size $t$ in the pseudo-trapezoid form. Without loss of generality, we assume that all cells of $\Xi$ intersect $L_{\leq k}(\calA)$. Define $Q$ to be the set of vertices of all cells of $\Xi$. Define $S$ to be the top edges of all cells of $\Xi$ that are segments of $\ell$. Since the interiors of cells of $\Xi$ are pairwise disjoint, the segments of $S$ are also interior pairwise-disjoint. As $\Xi$ has $t$ cells, we have $|Q|=O(t)$. In the following, we argue that $(Q,S)$ is a $(k,k+K)$-shallow cutting in the vertex-segment form. 

First of all, by definition, endpoints of all segments of $S$ are in $Q$.
Consider a point $q\in Q$, which is a vertex of a cell $\Delta\in \Xi$. As $\Delta$ intersects $L_{\leq k}(\calA)$ and $|\calA_{\Delta}|\leq K$, there are at most $k+K$ arcs of $\calA$ below $q$. Hence, $q$ has level at most $k+K$ in $\calA$.  
For each segment $s\in S$, since it is a top edge of a cell $\Delta\in \Xi$ and $|\calA_{\Delta}|\leq K$, we obtain $|\calA_s|\leq K$.

It remains to argue that $\vd(Q\cup S)$ covers $L_{\leq k}(\calA)$. By definition, the union of all cells of $\Xi$ covers $L_{\leq k}(\calA)$. Consider a cell $\Delta\in \Xi$, which is a pseudo-trapezoid. We show that $\Delta\subseteq \vd(Q\cup S)$, which will prove that $\vd(Q\cup S)$ covers $L_{\leq k}(\calA)$. 

Let $e$ be the top edge of $\Delta$. As $\Delta$ is a pseudo-trapezoid, $e$ is either a segment on $\ell$ or an $\bbR^+$-constrained arc. If $e$ is a segment of $\ell$, then $e\in S$ and thus $\Delta$ must be contained in a cell of $\vd(S)$. Hence, $\Delta\subseteq\vd(S)\subseteq \vd(Q\cup S)$. If $e$ is an $\bbR^+$-constrained arc, then let $q_1$ and $q_2$ be its two endpoints; thus $e$ is the arc $\gamma(q_1,q_2)$. Since $\Delta$ is a pseudo-trapezoid with $\gamma(q_1,q_2)$ as the top edge, $\Delta$ must be contained in the line-separated $\alpha$-hull $H_{\ell}(\{q_1,q_2\})$, which is a subset of $H_{\ell}(Q)$ by Observation~\ref{obser:hull}(1) as $q_1,q_2\in Q$. Recall that $\vd(Q)$ is the vertical decomposition of $H_{\ell}(Q)$.
Hence, we have $\Delta\subseteq \vd(Q)\subseteq \vd(Q\cup S)$.

This proves $\Delta\subseteq \vd(Q\cup S)$ and therefore  $\vd(Q\cup S)$ covers $L_{\leq k}(\calA)$.
\end{proof}

\begin{lemma}\label{lem:vertextoregular}
A $(k,K)$-shallow cutting of size $t$ in the vertex-segment form can be transformed into a $(k,3K)$-shallow cutting of size $O(t)$ in the bottom-open pseudo-trapezoid form. 
\end{lemma}
\begin{proof}
Let $(Q,S)$ be a $(k,K)$-shallow cutting of size $t$ in the vertex-segment form. We intend to take the vertical decomposition $\vd(Q,S)$ as the $(k,3K)$-shallow cutting $\Xi$ of size $O(t)$ in the bottom-open pseudo-trapezoid form. However, a subtle issue is that some cells of $\vd(Q,S)$ might not be pseudo-trapezoids. More specifically, consider a cell $\Delta\in \vd(Q,S)$. Let $e$ be the top edge of $\Delta$. According to the definition of $\vd(Q,S)$, $e$ belongs to one of the three cases: (1) $e$ is an $\bbR^+$-constrained arc; (2) $e$ is a segment of $\ell$; (3) $e$ is a segment of a wing half-line of a point of $Q$. In the first two cases, $\Delta$ is a bottom-open pseudo-trapezoid and we include $\Delta$ in $\Xi$. In the third case, $\Delta$ is not a pseudo-trapezoid by our definition since $e$ is a line segment but not on $\ell$. In this case, we extend $\Delta$ by moving $e$ upwards until $\ell$ to obtain an extended cell $\Delta'$, which is a bottom-open pseudo-trapezoid; we add $\Delta'$ to $\Xi$. We call $\Delta'$ a {\em special cell} of $\Xi$. 

We claim that $\calA_{\Delta'}=\emptyset$, i.e., $\Delta'$ does not intersect any arcs of $\calA$. Indeed, let $e'$ be the top edge of $\Delta'$. Let $R$ be the rectangular region of $\Delta'$ between $e$ and $e'$ (see Fig.~\ref{fig:extendcell}). Then, $\Delta'=\Delta\cup R$. Consider any point $p\in R$. We argue that no arc of $\calA$ is below $p$, which will prove the claim. Assume to the contradiction that there is an arc $\gamma\in \calA$ below $p$. Without loss of generality, we assume that $\gamma$ is the lowest arc of $\calA$ intersecting the vertical downward ray $\rho(p)$. Let $p'$ be the intersection of $\gamma$ and $\rho(p)$. By definition, $p'\in L_{\leq 0}(\calA)$. Since $(Q,S)$ is a $(k,K)$-shallow cutting, $\vd(Q,S)$ covers $L_{\leq k}(\calA)$ and thus covers $L_{\leq 0}(\calA)$ as $k\geq 0$. Therefore, $\vd(Q,S)$ covers $p'$. 
On the other hand, since $e$ is on a wing half-line of a point of $Q$, the $y$-coordinate of $e$ is $-1$ and thus no arcs of $\calA$ intersect $e$. Hence, $p'$ must be above $e$. But since $e$ is the top edge of the cell $\Delta\in \vd(Q,S)$, $p'$ cannot be covered by $\vd(Q,S)$, a contradiction. This proves that $\calA_{\Delta'}=\emptyset$.

\begin{figure}[t]
\begin{minipage}[t]{\linewidth}
\begin{center}
\includegraphics[totalheight=1.6in]{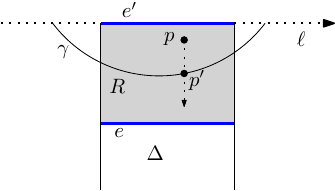}
\caption{Illustrating the proof of $\calA_{\Delta'}=\emptyset$, with $\Delta'=R\cup \Delta$, where $R$ is the gray rectangle and $\Delta$ is the region below $e$.}
\label{fig:extendcell}
\end{center}
\end{minipage}
\end{figure}

Since $|Q|=t$, $\vd(Q\cup S)$ has $O(t)$ cells. By definition, the size of $\Xi$ is $O(t)$. We next show that $\Xi$ is a $(k,3K)$-shallow cutting in the bottom-open pseudo-trapezoid form. 
First of all, by definition, each cell of $\vd(Q,S)$ is a bottom-open pseudo-trapezoid. Also, since $\vd(Q,S)$ covers $L_{\leq k}(\calA)$ and each cell of $\vd(Q,S)$ is either in $\Xi$ or contained in a cell of $\Xi$, $\vd(Q,S)$ is a subset of $\Xi$. Therefore, $\Xi$ covers $L_{\leq k}(\calA)$. 
For each $\Delta$ of $\Xi$, if it is a special cell, then $|\calA_{\Delta}|=0$ as proved above; otherwise 
$\Delta$ is also a cell in $\vd(Q,S)$ and we have 
$|\calA_{\Delta}|\leq 3K$ by Observation~\ref{obser:cellconlist}. Therefore, $\Xi$ is a $(k,3K)$-shallow cutting in the bottom-open pseudo-trapezoid form. 
\end{proof}

Combining Lemmas~\ref{lem:shallowexist10} and \ref{lem:regulartovertex} leads to the following. 
\begin{corollary}\label{coro:shallowexist20}
Given $k\in [1,n]$, there exists a $(k,O(k))$-shallow cutting of size $O(n/k)$ in the vertex-segment form.     
\end{corollary}

\subsection{Computing shallow cuttings in the vertex-segment form}
\label{subsec:algoshallowcut}

In what follows, by extending the algorithm of \cite{ref:ChanOp16}, we present an algorithm to compute shallow cuttings for $\calA$ in the vertex-segment form. 

We say that a (standard) cutting is in the {\em pseudo-trapezoid form} if every cell of the cutting is a pseudo-trapezoid. The following result was known previously~\cite{ref:ChazelleCu93,ref:WangUn23}. 
\begin{lemma}{\em \cite{ref:ChazelleCu93,ref:WangUn23}}\label{lem:algocut}
Given any constant $\epsilon>0$, an $\epsilon$-cutting for $\calA$ in the pseudo-trapezoid form of $O(1)$ size covering the plane, along with its conflict lists, can be computed in $O(n)$ time. 
\end{lemma}

We say that a shallow cutting $(Q,S)$ in the vertex-segment form is {\em sorted} if points of $Q$ are sorted by their $x$-coordinates. The following is the main theorem about our algorithm.  

\begin{theorem}\label{theo:algoshallowcut}
There exist constants $B,C,C'$, such that for any parameter $k\in [1,n]$, given a $(Bk,CBk)$-shallow cutting $(\iq,\is)$ in the sorted vertex-segment form for $\calA$ of size at most $C'\frac{n}{Bk}$ along with its conflict lists, we can compute a $(k,Ck)$-shallow cutting $(\oq,\os)$ in the sorted vertex-segment form for $\calA$ of size at most $C'\frac{n}{k}$ along with its conflict lists in $O(n)$ time. 
\end{theorem}
\begin{proof}
Let $\epsilon$ be a constant to be set later. 
We begin by computing the decomposition $\vd(\iq,\is)$. Since $\iq$ is sorted, $H_{\ell}(\iq)$ can be computed in $O(|\iq|)$ time using the algorithm from Section~\ref{sec:alphahull}. As the endpoints of all segments of $\is$ are in $\iq$ and segments of $\is$ are interior pairwise-disjoint on $\ell$, the segments of $\is$ can also be sorted from left to right in $O(|\iq|)$ time. As such, computing $\vd(\iq,\is)$ can be done in $O(|\iq|)$ time, which is $O(n/k)$ as $|\iq|\leq C'\frac{n}{Bk}$. 

Next, for each cell $\Delta\in \vd(\iq,\is)$, we perform the following two steps. 
\begin{enumerate}
    \item Compute an $\epsilon$-cutting $\Xi_{\Delta}$ of size $O(1)$ for $\calA_{\Delta}$. We clip the cells of $\Xi_{\Delta}$ to lie within $\Delta$ (and redecompose each new cell into pseudo-trapezoids if needed). Let $Q_{\Delta}$ denote the set of vertices of all cells of $\Xi_{\Delta}$ and $S_{\Delta}$ the set of top edges of the cells of $\Xi_{\Delta}$ that are segments of $\ell$. 
    
    Since $\epsilon=O(1)$, $\Xi_{\Delta}$ has $O(1)$ cells and computing $\Xi_{\Delta}$ takes $O(|\calA_{\Delta}|)$ time by Lemma~\ref{lem:algocut}. Hence, both $|Q_{\Delta}|$ and $|S_{\Delta}|$ are $O(1)$. As $\sum_{\Delta\in \vd(\iq,\is)}|\calA_{\Delta}|=O(n)$, the total time of this step for all cells $\Delta\in \vd(\iq,\is)$ is $O(n)$. 
    
    \item Compute by brute force a smallest subset $Q'_{\Delta}\subseteq Q_{\Delta}$, along with 
    a subset $S'_{\Delta}\subseteq S_{\Delta}$, such that the following conditions are satisfied. 
    \begin{enumerate}
        \item The endpoints of all segments of $S'_{\Delta}$ are in $Q'_{\Delta}$.
        \item Every vertex in $Q'_{\Delta}$ has level in $\calA_{\Delta}$ at most $Ck$. 
        \item Every segment in $S'_{\Delta}$ intersects at most $Ck$ arcs of $\calA_{\Delta}$.
        \item For each cell $\sigma\in \Xi_{\Delta}$ whose vertices are all in $L_{\leq 2k}(\calA_{\Delta})$,  $\sigma$ is covered by $\vd(Q'_{\Delta},S'_{\Delta})$.
     \end{enumerate}
    

    As both $|Q_{\Delta}|$ and $|S_{\Delta}|$ are $O(1)$, there are $O(1)$ different pairs of $Q'_{\Delta}$ and $S'_{\Delta}$. For each such pair $(Q'_{\Delta},S'_{\Delta})$, we can check whether the four conditions are satisfied in $O(|\calA_{\Delta}|)$ time, because $|Q'_{\Delta}|$, $|S'_{\Delta}|$, and the size of $\Xi_{\Delta}$ are all $O(1)$. Hence, finding a smallest subset $Q'_{\Delta}$ with $S'_{\Delta}$ takes $O(|\calA_{\Delta}|)$ time. 
    After that, for each point $q\in Q'_{\Delta}$, its conflict list in $\calA_{\Delta}$, which is also its conflict list in $\calA$, can be found in $O(|\calA_{\Delta}|)$ time. Similarly, for each segment of $S'_{\Delta}$, its conflict list can be found in $O(|\calA_{\Delta}|)$ time. As both $|Q'_{\Delta}|$ and $|S'_{\Delta}|$ are $O(1)$, finding the conflict lists of $(Q'_{\Delta},S'_{\Delta})$ takes $O(|\calA_{\Delta}|)$ time. 
    
    As $\sum_{\Delta\in \vd(\iq,\is)}|\calA_{\Delta}|=O(n)$, the total time of this step for all $\Delta\in \vd(\iq,\is)$ is $O(n)$. 
   
\end{enumerate}
    
We define $\oq=\bigcup_{\Delta\in \vd(\iq,\is)}Q'_{\Delta}$ and $\os=\bigcup_{\Delta\in \vd(\iq,\is)}S'_{\Delta}$. We can sort the points of $\oq$ in $O(n/k)$ time as follows. For each $\Delta \in\vd(\iq,\is)$, we sort the points of $Q'_{\Delta}$, which takes $O(1)$ time as $|Q'_{\Delta}|=O(1)$. Then, for all cells $\Delta \in\vd(\iq,\is)$ in order from left to right, we concatenate the sorted lists of $Q'_{\Delta}$, and the resulting list is a sorted list of $\oq$. This takes $O(n/k)$ time in total as $\vd(\iq,\is)$ has $O(n/k)$ cells. 

The total runtime of the above algorithm is $O(n)$. 

\paragraph{Correctness.}
In the following, we argue the correctness, i.e., prove that $(\oq,\os)$ is a $(k,Ck)$-shallow cutting for $\calA$ of size at most $C'\frac{n}{k}$. We first show that $(\oq,\os)$ is a $(k,Ck)$-shallow cutting and then bound its size. 

Consider a cell $\Delta \in\vd(\iq,\is)$. For each point $q\in Q'_{\Delta}$, according to our algorithm, $q$ has level in $\calA_{\Delta}$ at most $Ck$. In light of the definition of $\vd(\iq,\is)$, $\Delta$ is a bottom-open cell bounded by two vertical rays. Hence, the level of $q$ in $\calA_{\Delta}$ is also its level in $\calA$. Therefore, $q$ has level in $\calA$ at most $Ck$. 

For each segment $s\in S'_{\Delta}$, by definition, $s$ is a top edge of a cell $\Delta\in \vd(\iq,\is)$. According to our algorithm, $s$ intersects at most $Ck$ arcs of $\calA_{\Delta}$. Since $s\subseteq \Delta$, any arc of $\calA$ intersecting $s$ must intersect $\Delta$ and thus is in $\calA_{\Delta}$. Therefore, $s$ intersects at most $Ck$ arcs of $\calA$. In addition, according to our algorithm, the endpoints of $s$ are in $Q'_{\Delta}$. 

To show that $(\oq,\os)$ is a $(k,Ck)$-shallow cutting, it remains to prove that $\vd(\oq,\os)$ covers $L_{\leq k}(\calA)$. Consider a point $p\in L_{\leq k}(\calA)$. By definition, $\vd(\iq,\is)$ covers $L_{\leq Bk}(\calA)$. By setting $B>1$, $\vd(\iq,\is)$ covers $L_{\leq k}(\calA)$ and thus $p$ must be in a cell $\Delta$ of $\vd(\iq,\is)$. In the following, we argue that $p$ is covered by $\vd(Q'_{\Delta}\cup S'_{\Delta})$, which will prove that $\vd(\oq,\os)$ covers $L_{\leq k}(\calA)$. 

Let $\sigma$ be the cell of the cutting $\Xi_{\Delta}$ that contains $p$. Since $p$ has level in $\calA_{\Delta}$ at most $k$ and the number of arcs of $\calA_{\Delta}$ intersecting $\sigma$ is at most $\epsilon\cdot |\calA_{\Delta}|$, every vertex of $\sigma$ has level at most $k+\epsilon\cdot |\calA_{\Delta}|$. Because the size of the conflict list of each point of $\iq$ is at most $CBk$, $|\calA_{\Delta}|\leq 3CBk$ by Observation~\ref{obser:cellconlist}. Therefore, $k+\epsilon\cdot |\calA_{\Delta}|\leq 2k$ by setting the constant $\epsilon = 1/(3CB)$. Hence, all vertices of $\sigma$ are in $L_{\leq 2k}(\calA_{\Delta})$ and thus $\sigma$ is covered by $\vd(Q'_{\Delta}\cup S'_{\Delta})$ according to our algorithm. 
As $p\in \sigma$, we obtain that $p$ is covered by $\vd(Q'_{\Delta}\cup S'_{\Delta})$.


\paragraph{Bounding the size of $\vd(\oq,\os)$, i.e., $|\oq|$.} We now prove $|\oq|\leq C'n/k$. To this end, we compare it against a $(5k,5c_0k)$-shallow cutting $(Q^*,S^*)$ of size at most $c_0'n/(5k)$ for some constant $c_0'$, whose existence is guaranteed by Corollary~\ref{coro:shallowexist20}. 

We set constant $B\geq 15c_0$. We first claim that $\vd(\iq,\is)$ covers $\vd(Q^*,S^*)$. Indeed, 
consider a cell $\Delta\in \vd(Q^*,S^*)$. 
By Observation~\ref{obser:cellconlist}, $\Delta$ intersects at most $15c_0k$ arcs of $\Gamma$. Hence, $\Delta\subseteq L_{\leq 15c_0k}(\Gamma)$.  As $\vd(\iq,\is)$ covers $L_{\leq Bk}(\calA)$ and $B\geq 15c_0$, $\Delta\subseteq\vd(\iq,\is)$. Hence, $\vd(\iq,\is)$ covers $\vd(Q^*,S^*)$. 


We render $(Q^*,S^*)$ comparable to $(\oq,\os)$ by performing three modification steps: (1) Add some points to $Q^*$; (2) modify $S^*$; (3) remove some points of $Q^*$. Note that these modification steps are not part of the algorithm but for this proof only. 

\paragraph{Step (1): Adding points to $\boldsymbol{Q^*}$.}
We chop $\vd(Q^*,S^*)$ at the walls (i.e., the edges) of the cells of $\vd(\iq,\is)$. Consider a cell $\Delta\in \vd(\iq,\is)$. Consider the right edge of $\Delta$, which is a vertical downward ray $\rho(q)$ from the right endpoint $q$ of the top edge of $\Delta$. Since $\vd(\iq,\is)$ covers $\vd(Q^*,S^*)$ and $\partial \vd(Q^*,S^*)$ is $x$-monotone, an edge $e$ of $\partial \vd(Q^*,S^*)$ must intersect $\rho(q)$ at a point $p$. 
We create a new vertex at $p$ for $\vd(Q^*,S^*)$ and add two copies of it to $Q^*$ (one assigned to each of the two incident cells of $\vd(Q^*,S^*)$). Since each endpoint of $e$ has level at most $5c_0k$, the new vertex $p$ has level at most $15c_0k$ by Observation~\ref{obser:cellconlist}. In this way, the number of extra vertices added is at most $2|\vd(\iq,\is)|$, where $|\vd(\iq,\is)|$ is the number of cells of $\vd(\iq,\is)$. Since $|\iq|\leq C'\frac{n}{Bk}$, we have $|\vd(\iq,\is)|\leq 5C'\frac{n}{Bk}$. Thus the size of $Q^*$ is now at most $(\frac{c_0'}{5}+10\frac{C'}{B})\frac{n}{k}$.

For each cell $\Delta\in \vd(\iq,\is)$, define $Q^*_{\Delta}=Q^*\cap \Delta$. For each cell $\sigma\in \Xi_{\Delta}$, if a point $p$ of $Q^*_{\Delta}$ is in $\sigma$, then we snap $p$ to the at most four vertices of $\sigma$ by adding the vertices of $\sigma$ to $Q^*_{\Delta}$ (we refer to it as the {\em snapping process}; $p$ is still kept in $Q^*_{\Delta}$ for now but will be deleted later in the third modification step). After this, the coverage of $\vd(Q^*_{\Delta})$ only increases. 
Every point in $Q^*_{\Delta}$ now has level at most $15c_0k+\epsilon\cdot |\calA_{\Delta}|\leq 15c_0k+\epsilon(3CBk)=(15c_0+1)k$ since $\epsilon=1/(3CB)$. The size of $Q^*$ becomes at most five times larger, which is at most $5(\frac{c_0'}{5}+10\frac{C'}{B})\frac{n}{k}=(c_0'+50\frac{C'}{B})\frac{n}{k}$. 

\paragraph{Step (2): Modifying $\boldsymbol{S^*}$.}
We now modify $S^*$. Consider a segment $s\in S^*$. By definition, $s\subseteq \ell$. If the interior of $s$ contains a vertex $q$ of a cell of $\vd(\iq,\is)$, we break $s$ at $q$ into two segments and use them to replace $s$ in $S^*$. 
If we process every segment of $S^*$ as above, we obtain a new set $S^*$ such that $\vd(S^*)$ has the same coverage as before (i.e., the union of the new segments is the same as that of the old segments), but every segment of $S^*$ is now in a single cell of $\vd(\iq,\is)$. For each cell $\Delta\in \vd(\iq,\is)$, we use $S^*_{\Delta}$ to denote the subset of segments of $S^*$ that are in $\Delta$. In addition, since each new segment $s$ is a sub-segment of an old segment, which intersects at most $5c_0k$ arcs of $\calA$ by definition, $s$ intersects at most $5c_0k$ arcs of $\calA$. 

For each cell $\Delta\in \vd(\iq,\is)$, we further adjust the segments of $S^*_{\Delta}$ to obtain a new set $S^*_{\Delta}$ such that (1) the coverage of the old set is a subset of the new coverage; (2) all segment endpoints are in $Q^*_{\Delta}$ and also in $Q_{\Delta}$, i.e., in $Q^*_{\Delta}\cap Q_{\Delta}$; (3) each segment intersects at most $(5c_0+2)k$ segments of $\calA$. 

We process the segments of $S^*_{\Delta}$ from left to right. Consider the leftmost segment $s_1$. Let $l_1$ and $r_1$ be the left and right endpoints of $s_1$, respectively (see Fig.~\ref{fig:segment}). Since $l_1\in \ell$, $l_1$ must be on the top edge $e_{l1}$ of a cell of $\Xi_{\Delta}$ and $e_{l1}\subseteq \ell$. Let $a_{l1}$ and $b_{l1}$ be the left and right endpoints of $e$, respectively. By definition, $a_{l1}$ and $b_{l1}$ are in $Q_{\Delta}$.
By the above snapping process, both $a_{l1}$ and $b_{l1}$ are already in $Q^*_{\Delta}$. 
Similarly, we define $e_{r1}$, $a_{r1}$, and $b_{r1}$ for $r_1$, and both $a_{r1}$ and $b_{r1}$ are in $Q_{\Delta}\cap Q^*_{\Delta}$. 
We replace the original segment $s_1$ with the new segment $\overline{a_{l1}b_{r1}}$. Clearly, the old segment is a subset of the new segment, i.e., the coverage increases, and both endpoints of the new segment are now in $Q_{\Delta}\cap Q^*_{\Delta}$. The number of arcs intersecting the new segment is at most the sum of the following: (1) the number of arcs intersecting $e_{l1}$, which is $\epsilon\cdot |\calA_{\Delta}|\leq k$; (2) the number of arcs intersecting $e_{r1}$, which is also $\epsilon\cdot |\calA_{\Delta}|\leq k$; (3) the number of arcs intersecting the original segment, which is $5c_0k$. Therefore, the number of arcs of $\calA$ intersecting the new segment is $(5c_0+2)k$. 

\begin{figure}[t]
\begin{minipage}[t]{\linewidth}
\begin{center}
\includegraphics[totalheight=0.8in]{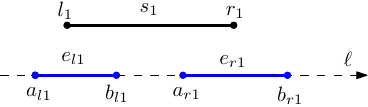}
\caption{Illustrating the notation for processing $s_1$. The segment $s_1$ is supposed be on $\ell$, but we move it upwards off $\ell_1$ for better illustration.}
\label{fig:segment}
\end{center}
\end{minipage}
\end{figure}

Processing $s_2$ is done in a similar but slightly different way. Let $l_2$ and $r_2$ be the left and right endpoints of $s_2$, respectively. Define $e_{l2}$, $a_{l2}$, $b_{l2}$ for $l_2$, and $e_{r2}$, $a_{r2}$, $b_{r2}$ for $r_2$, similarly to the above. If $e_{r1}$ is not $e_{l2}$, then we process $s_2$ in the same way as above (i.e., replace $s_2$ by $\overline{a_{l2}b_{r2}}$). Otherwise, $e_{l2}$ is already included in the new segment for $s_1$; in this case, we replace $s_2$ by $\overline{b_{l2}b_{r2}}$ (see Fig.~\ref{fig:segment20}). Although the new segment does not fully cover the old $s_2$, the union of the new $s_1$ and the new $s_2$ still covers the union of the old $s_1$ and $s_2$. Also, both endpoints of the new $s_2$ are now in $Q_{\Delta}\cap Q^*_{\Delta}$. Following a similar argument to the above, the number of arcs of $\calA$ intersecting the new $s_2$ is at most $(5c_0+2)k$.

\begin{figure}[t]
\begin{minipage}[t]{\linewidth}
\begin{center}
\includegraphics[totalheight=1.0in]{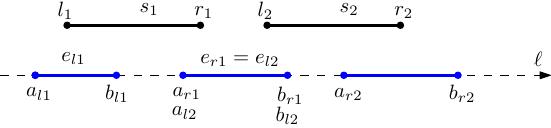}
\caption{Illustrating the notation for processing $s_2$ when $e_{r1}=e_{l2}$, i.e., $a_{r1}=a_{l2}$ and $b_{r1}=b_{l2}$.}
\label{fig:segment20}
\end{center}
\end{minipage}
\end{figure}

We process other segments in the same way as above for $s_2$. After all segments are processed, we will obtain a new set $S^*_{\Delta}$ that satisfies the three conditions mentioned above. 

Finally, we add a few more segments to $S^*$ as follows. For each cell $\sigma\in \Xi_{\Delta}$ whose top edge $e$ is on $\ell$ such that $\sigma$ contains a point $p\in Q_{\Delta}^*$, we do the following. By the above snapping process, all vertices of $\sigma$ are in $Q^*_{\Delta}$. In particular, both endpoints of $e$ are in $Q^*_{\Delta}$ (and also in $Q_{\Delta}$). According to our above way of modifying $S^*_{\Delta}$, either $e$ is contained in a segment of $S^*_{\Delta}$ or the interior of $e$ does not intersect any segment of $S^*_{\Delta}$. In the latter case, we add $e$ to $S^*$. This increases the coverage of $S^*$. 
As $e\subseteq\sigma$, 
the number of arcs intersecting the new segment $e$ is at most $\epsilon\cdot |\calA_{\Delta}|\leq k$. 

This finishes the modification of $S^*$. Each segment of $S^*$ intersects at most $(5c_0+2)k$ arcs of $\calA$. For each cell $\Delta\in \vd(\iq,\is)$, the endpoints of all segments of $S^*_{\Delta}$ are in $Q_{\Delta}\cap Q^*_{\Delta}$.

\paragraph{Step (3): Removing points of $\boldsymbol{Q^*}$.}
For each cell $\Delta\in \vd(\iq,\is)$, we do the following. It is possible that $Q^*_{\Delta}$ has points that are not in $Q_{\Delta}$. In this last step, we remove all those points from $Q^*_{\Delta}$. After this removal step, we have $Q^*_{\Delta}\subseteq Q_{\Delta}$. Since the endpoints of all segments of $S^*_{\Delta}$ are in $Q_{\Delta}\cap Q^*_{\Delta}$, none of the segment endpoint is removed from $Q^*_{\Delta}$. 

We claim that the coverage of $\vd(Q^*_{\Delta},S^*_{\Delta})$ does not change after the removal. Indeed, let $p$ be a point removed above. We argue that removing $p$ does not change $\vd(Q^*_{\Delta},S^*_{\Delta})$. 
Let $\sigma$ be the cell of $\Xi_{\Delta}$ that contains $p$.  Let $e$ be the top edge of $\sigma$. As $\sigma$ is a pseudo-trapezoid, $e$ is either a segment of $\ell$ or an $\bbR^+$-constrained arc. 
\begin{enumerate}
\item 
If $e$ is a segment, according to our modification on $S^*$, $e$ must be contained in a segment of $S^*$, implying that $p$ must be below (or on) a segment of $S^*$ (but $p$ is not a segment endpoint since every segment endpoint is in $Q_{\Delta}$ and $p\not\in Q_{\Delta}$). Hence, removing $p$ does not change the coverage of $\vd(S^*)$ and thus does not change the coverage of $\vd(Q^*_{\Delta},S^*_{\Delta})$. 

\item 
If $e$ is an $\bbR^+$-constrained arc, let $q_1$ and $q_2$ be the two endpoints of $e$. 
Thus $e$ is the arc $\gamma(q_1,q_2)$. As discussed in the proof of Lemma~\ref{lem:regulartovertex}, $\Delta\subseteq H_{\ell}(\{q_1,q_2\})$. Since $p\in \sigma$, by the snapping process, $q_1,q_2\in Q^*_{\Delta}$. By Observation~\ref{obser:hull}(1), $H_{\ell}(\{q_1,q_2\})\subseteq \vd(Q^*_{\Delta})\subseteq\vd(Q^*_{\Delta},S^*_{\Delta})$. On the other hand, since $p\not\in Q_{\Delta}$, $p$ cannot be a vertex of $\sigma$.
Thus, $p\not\in \{q_1,q_2\}$. Hence, removing $p$ does not change $H_{\ell}(\{q_1,q_2\})$ and thus does not change $\vd(Q^*_{\Delta})$ or  $\vd(Q^*_{\Delta},S^*_{\Delta})$. 
\end{enumerate}

In summary, the above removal step does not change the coverage of $\vd(Q^*_{\Delta},S^*_{\Delta})$. But now we have $Q^*_{\Delta}\subseteq Q_{\Delta}$. The endpoints of all segments of $S^*_{\Delta}$ are still in $Q^*_{\Delta}$.

\paragraph{Bounding $\boldsymbol{|\oq|}$.}
By setting $C=15c_0+2$, for each $\Delta\in \vd(\iq,\is)$, we claim that $(Q^*_{\Delta},S^*_{\Delta})$ satisfy the four conditions for $Q'_{\Delta}$ and $S'_{\Delta}$ in the second step of our algorithm. Indeed, the above already shows that the endpoints of all segments of $S^*_{\Delta}$ are in $Q^*_{\Delta}$; the first condition is thus satisfied. As discussed above, every point in $Q^*_{\Delta}$ has level at most $(15c_0+1)k<Ck$; the second condition is thus satisfied. The above also shows that every segment of $S^*_{\Delta}$ intersects at most $(5c_0+2)k<Ck$ arcs, which satisfies the third condition. For the fourth condition, consider a cell $\Delta\in \vd(\iq,\is)$ and let $\sigma$ be a cell of $\Xi_{\Delta}$ whose vertices are all in $L_{\leq 2k}(\calA_{\Delta})$. Consider any point $p\in \sigma$. We argue that $p$ is covered by $\vd(Q^*_{\Delta},S^*_{\Delta})$, which will prove the fourth condition. 

Let $e$ be the top edge of $\sigma$. An arc of $\calA$ is below $p$ only if it is below at least one endpoint of $e$ or intersects $e$. Hence, the number of arcs of $\calA$ below $p$ is at most the sum of the following two values: (1) the number of arcs of $\calA$ below at least one  endpoint of $e$, which is at most $4k$ since both points are in $L_{\leq 2k}(\calA_{\Delta})$; (2) the number of arcs intersecting $e$, which is at most $\epsilon\cdot |\calA_{\Delta}|\leq k$. Therefore, $p$ has level in $\calA$ is at most $5k$ and thus by definition is covered by $\vd(Q^*,S^*)$ of the original $(Q^*,S^*)$ before modification. After the modification of $Q^*$ and $S^*$, the coverage of $\vd(Q^*\cup S^*)$ increases. Hence, $\vd(Q^*,S^*)$ of the new $(Q^*,S^*)$ also covers $L_{\leq 5k}(\calA)$ and thus covers $p$. Our modification also guarantees that $\vd(Q^*,S^*)\cap \Delta=\vd(Q^*_{\Delta},S^*_{\Delta})$. Therefore, $p$ is covered by $\vd(Q^*_{\Delta},S^*_{\Delta})$. 
The fourth condition thus holds. 


As $(Q^*_{\Delta},S^*_{\Delta})$ satisfies the four conditions, we have $|Q'_{\Delta}|\leq |Q^*_{\Delta}|$ by the definition of $Q'_{\Delta}$. Summing over all cells in $\vd(\iq,\is)$ leads to $|\oq|\leq |Q^*|\leq (c_0'+50\frac{C'}{B})\frac{n}{k}$, which is less than or equal to $C'\frac{n}{k}$ as desired by setting the constant $C'=c_0'/(1-\frac{50}{B})$ with any constant $B>50$ (combining the above, we set $B>\max\{15c_0,50\}$). 

This proves the theorem. 
\end{proof}

Theorem~\ref{theo:algoshallowcut} leads to the following corollary. 
\begin{corollary}\label{coro:shallowcut}
There exist constants $B$, $C$, and $C'$, such that for any parameter $k\in [1,n]$, we can compute a $(B^ik,CB^ik)$-shallow cutting in the sorted vertex-segment form of size at most $C'\frac{n}{B^ik}$, along with its conflict lists, for all $i=0,1,\ldots,\log_B\frac{n}{k}$ in $O(n\log \frac{n}{k})$ total time. In particular, we can compute a $(k,Ck)$-shallow cutting of size $O(n/k)$ in the sorted vertex-segment form, along with its conflict lists, in $O(n\log \frac{n}{k})$ time.   
\end{corollary}
\begin{proof}
By Theorem~\ref{theo:algoshallowcut}, the runtime $T(n,k)$ satisfies the recurrence $T(n,k)=T(n,Bk)+O(n)$ with the trivial base case $T(n,n)=O(n)$. The recurrence solves to $T(n,k)=O(n\log \frac{n}{k})$. 
\end{proof}

\paragraph{Proving Theorem~\ref{theo:shallowcut}.}
We first apply Corollary~\ref{coro:shallowcut} to compute the shallow cuttings in the sorted vertex-segment form. Then, we transform them to shallow cuttings in the bottom-open pseudo-trapezoid form by Lemma~\ref{lem:vertextoregular}, which can be done in additional $O(n\log\frac{n}{k})$ time (i.e., linear time for each cutting). This proves Theorem~\ref{theo:shallowcut}.

\section{The static UDRR problem}
\label{sec:udrr}
In this section, we discuss our static data structure. 
Let $P$ be a set of $n$ points in the plane. The problem is to construct a data structure for $P$ to answer unit-disk range reporting queries.

In the preprocessing, we apply the algorithm of Lemma~\ref{lem:grid} to compute a conforming coverage set $\calC$ of $O(n)$ cells for $P$ and the data structure of Lemma~\ref{lem:grid}(2).

Consider a query unit disk $D_q$ whose center is $q$. If $q$ is not in a cell of $\calC$, then  $P(D_q)=\emptyset$ and we return null. Otherwise, as discussed in Section~\ref{sec:dynamicreport}, it suffices to
report $P(C') \cap D_q$ for all cells $C' \in N(C)$. If $C' = C$,
we report all points of $P(C)$. Otherwise, $C$ and $C'$ are separated by an axis-parallel line. Without loss of generality, we assume that $C$ and $C'$ are separated by a horizontal line $\ell$ with $C'$ above $\ell$ and $C$ below $\ell$. As $q\in C$, $q$ is separated from $C'$ by $\ell$.
Our goal is to report points of $P(C')\cap D_q$. We formulate the problem as the following {\em static line-separable UDRR problem}:

\begin{problem}(Static line-separable UDRR)
\label{problem:LS-UDRR-old}
Given a set $Q$ of $m$ points above a horizontal line $\ell$ such that all
points of $Q$ are contained in a unit disk, build a data structure so that for
any query unit disk $D_q$ centered at a point $q$ below $\ell$, the
points of $Q$ in $D_q$ can be reported efficiently.
\end{problem}

Note that comparing to Problem~\ref{problem:dynamic-LS-UDRR}, we require all points of $Q$ to be contained in a unit disk. This constraint, which suffices for our purpose, will make our algorithm simpler. To solve our problem, we can set $Q=P(C')$ since all points of $P(C')$ are in $C'$, which is contained in a unit disk. In what follows, we will prove Lemma~\ref{lem:lineUDRR}.

\begin{lemma}\label{lem:lineUDRR}
For the static line-separable UDRR, we can build a data structure of $O(m)$ space in
$O(m\log m)$ time that can answer each query in $O(k+\log m)$ time,
where $k$ is the output size.
\end{lemma}

Before proving Lemma~\ref{lem:lineUDRR}, we prove the following main result for UDRR using Lemma~\ref{lem:lineUDRR}.

\begin{theorem}
    \label{theo:udrr}
    Given a set $P$ of $n$ points in the plane, we can build a data structure of $O(n)$ space in $O(n \log n)$ time such that given any query unit disk, the points of $P$ in the disk can be reported in $O(\log n + k)$ time, where $k$ is the output size.
\end{theorem}
\begin{proof}
    We first compute a conforming coverage set $\calC$ of $O(n)$ cells for $P$ and build the data structure $\calD$ of  Lemma~\ref{lem:grid}(2). Then, for each cell $C\in \calC$ that contains at least one point of $P$, we construct a data structure $\calD_e(C)$ of Lemma~\ref{lem:lineUDRR} for $P(C)$ with respect to the supporting line of each edge $e$ of $C$, which takes $O(|P(C)|)$ space and $O(|P(C)|\cdot \log |P(C)|)$ time. Since each cell of $\calC$ has four edges and $\sum_{C\in \calC}|P(C)|=n$, the total space of the overall data structure is $O(n)$ and the total preprocessing time is $O(n\log n)$. 

    Given a query unit disk $D_q$ with center $q$, we first check whether $q$ is in a cell of $\calC$, and if so, find such a cell; this takes $O(\log n)$ time by Lemma~\ref{lem:grid}. If no cell of $\calC$ contains $q$, then $P\cap D_q=\emptyset$ and we simply return null. Otherwise, let $C$ be the cell of $\calC$ that contains $q$. We first report all points of $P(C)$. Next, for each $C'\in N(C)$, by Definition~\ref{def:cell}(3), $C$ and $C'$ are separated by an axis-parallel line $\ell$. Since each edge of $C$ and $C'$ is axis-parallel, $C'$ must have an edge $e$ whose supporting line is parallel to $\ell$ and separates $C$ from $C'$. Using  $\calD_e(C')$, we report all points of $P(C')$ inside $D_q$. As $|N(C)|=O(1)$, the total query time is $O(\log n + k)$ by Lemma~\ref{lem:lineUDRR}.
\end{proof}

\subsection{Static line-separable UDRR: Proving Lemma~\ref{lem:lineUDRR}}
\label{sec:subproblem}

We now prove Lemma~\ref{lem:lineUDRR}. For convenience, we use $n$ to denote the size of $Q$ instead of $m$. 

Consider a query unit disk $D_q$ whose center $q$ is below $\ell$. The goal of
the query is to report $Q(D_q)$.
As in Section~\ref{sec:lemdynamiclineUDRR} for the dynamic problem, we define $\calA$ as the set of arcs below $\ell$ of the circles centered at the points of $Q$. 
Reporting the points of $Q$ in $D_q$ becomes reporting the arcs of $\calA$ below $q$.

Define $\calU_1$ as the lower envelope of the arcs of $\calA$ (see
Fig.~\ref{fig:OneLowerEnvelope}). Since each arc of $\calA$ is $x$-monotone,
$\calU_1$ is also $x$-monotone. Note that $\calU_1$ may have several connected
components (see Fig.~\ref{fig:OneLowerEnvelopeSeveralComponents}). Observe that
$q$ is above an arc of $\calA$ if and only if $q$ is above $\calU_1$ (see
Fig.~\ref{fig:OneLowerEnvelope}). It has
been proved by Wang and Zhao (Lemma 9 in~\cite{ref:WangCo22}) that each arc of $\calA$ can
contribute at most one arc in $\calU_1$. Suppose we traverse arcs of $\calU_1$
from left to right; the order of these arcs encountered during our traversal is
called the \emph{traversal order}.
The following lemma shows that the traversal order is consistent with the order of the arcs of $\calU_1$ sorted by their centers from left to right.

%

 \begin{figure}[t]
     \centering
     \begin{minipage}[t]{0.48\textwidth}
         \centering
         \includegraphics[height=1.2in]{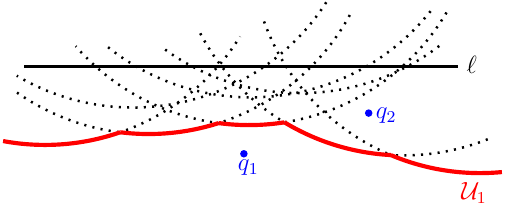}
         \caption{Illustrating the lower envelope $\calU_1$. Black dotted arcs are boundaries of unit disks centered at points of $Q$. The point $q_1$ is below $\calU_1$ while $q_2$ is above $\calU_1$.}
         \label{fig:OneLowerEnvelope}
     \end{minipage}
     \hspace{0.08in}
     \begin{minipage}[t]{0.48\textwidth}
         \centering
         \includegraphics[height=1.2in]{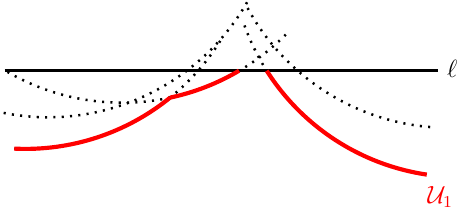}
         \caption{Illustrating a lower envelope $\calU_1$ with two connected components.}
         \label{fig:OneLowerEnvelopeSeveralComponents}
     \end{minipage}
 \end{figure}


\begin{lemma}
    \label{lem:CenterOrderAndTraversalOrder}
    The centers of arcs in $\calU_1$ following the traversal order are sorted in ascending order by $x$-coordinate.
\end{lemma}
\begin{proof}
    \label{proof:lem-CenterOrderAndTraversalOrder}
    Consider two consecutive arcs $\calU_1$ following the traversal order and
	let $\gamma_i$ and $\gamma_{i+1}$ be the two arcs of $\calA$ containing them, respectively.
    Let $p_i$ and $p_{i+1}$ are the centers of $\gamma_i$ and $\gamma_{i + 1}$, respectively. Our goal is to prove that $x(p_i)\leq x(p_{i+1})$.

    Let $a_j$ and $b_j$ be the left and right endpoints of $\gamma_j$, for $j\in \{i,i+1\}$. It has been proved in Lemma 9 of~\cite{ref:WangCo22} that
		$x(b_i)\leq x(b_{i+1})$ because the subarc of $\gamma_i$ on $\calU_1$ appears
		in the front of that of $\gamma_{i+1}$ in the traversal order; symmetrically, $x(b_i)\leq x(b_{i+1})$ also holds.
  As such, we obtain $x(p_i)\leq x(p_{i+1})$ since $x(p_j)=(x(a_j)+x(b_j))/2$ for $j\in \{i,i+1\}$.
\end{proof}



Define $Q_1$ as the set of centers of all arcs of $\calU_1$.
We say that an arc $\gamma$ of $\calU_1$ {\em spans} a point $p$, if $x(p)$ is
between the $x$-coordinates of the two endpoints of $\gamma$.

\begin{lemma}
    \label{lem:SearchInOneLayer}
    Suppose $q$ is a point below $\ell$ and the arc of $\calU_1$ spanning $q$ is
	known; then the points of $Q_1\cap D_q$ can be reported
	in $O(|Q_1\cap D_q|)$ time (assuming that $\calU_1$ is stored in a data
	structure so that one can access from each arc of $\calU_1$ its
	neighboring arcs in $O(1)$ time).
\end{lemma}
\begin{proof}
    \label{proof:lem-SearchInOneLayer}
    Let $\gamma'_1,\gamma'_2,\ldots,\gamma'_t$ be the arcs of $\calU_1$ following their
	traversal order, where $t$ is the number of arcs of $\calU_1$. For each
	$1\leq i\leq t$, let $p_i$ be the center of $\gamma_i'$ and $\gamma_i$ be the arc of
	$\calA$ containing $\gamma_i'$. By definition, $Q_1=\{p_1,p_2,\ldots, p_t\}$.

	If no arc of $\calU_1$ spans $q$, then it is not difficult to see that
	$Q_1\cap D_q=\emptyset$. In the following, we assume that $\calU_1$ has an arc
	spanning $q$, denoted by $\gamma_i'$.

	If $q$ is below $\gamma_i'$, then $q$ is below $\calU_1$ and thus $Q_1\cap
	D_q=\emptyset$. We thus assume that $q$ is above $\gamma_i'$ (see Fig.~\ref{fig:SearchInOneLowerEnvelope}, where $\gamma_i'$ is $\gamma_3'$). In this case, $p_i$
	is in $D_q$ and we report it. Next, starting from $\gamma_i'$, we traverse
	on the arcs of $\calU_1$ rightwards (resp., leftwards) until the distance
	between $q$ and the center of an arc is larger than $1$.
    Specifically, for the rightwards case, we check the arcs of $\{\gamma'_{i + 1},
	\gamma'_{i + 2}, ...\}$ in this order and for each arc $\gamma_j'$, $j\geq i+1$, if $p_j$ is in
	$D_q$, then we report $p_j$ and proceed on $j+1$; otherwise, we halt the procedure. The leftwards
	case is symmetric.
    To see the correctness, we only argue the rightwards case as the other case is symmetric.

 \begin{figure}[t]
     \centering
     \begin{minipage}[t]{0.48\textwidth}
         \centering
         \includegraphics[height=1.3in]{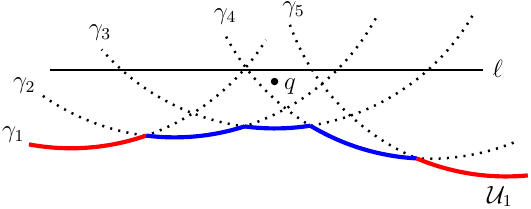}
         \caption{The three blue arcs are below $q$ while the two red arcs are above $q$.}
         \label{fig:SearchInOneLowerEnvelope}
     \end{minipage}
     \hspace{0.08in}
     \begin{minipage}[t]{0.48\textwidth}
         \centering
         \includegraphics[height=1.3in]{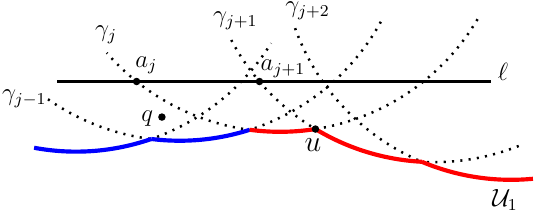}
         \caption{Illustrating the case where $\gamma_j'$ and $\gamma_{j+1}'$ intersect at a vertex $u$ of $\calU_1$.}
         \label{fig:intersect}
     \end{minipage}
 \end{figure}

    Suppose $p_j$ is outside $D_q$. Our goal is to show that $p_h$ is not in
	$D_q$ for any $j+1\leq h\leq t$.
    Consider the arc $\gamma'_{j + 1}$. There are two cases depending on whether
	$\gamma'_j$ and $\gamma'_{j + 1}$ intersect. Let $a_i$ and $b_i$ be the left and right
	endpoints of $\gamma_i$, respectively,
	for $i\in \{j,j+1\}$.

	\begin{itemize}
	\item
    If $\gamma'_j$ and $\gamma'_{j + 1}$ intersect, say, at a point $u$, then $u$ is a
	vertex of $\calU_1$ (see Fig.~\ref{fig:intersect}). As $q$ is spanned by $\gamma_i'$ and $i<j$, it holds that
	$x(q)<x(u)$. Since $p_j$ is outside $D_q$, $q$ is not above $\gamma_j$, and more
	specifically, not above the portion of $\gamma_j$ between $a_j$ and $u$. Since
	$\gamma'_j$ and $\gamma'_{j + 1}$ intersect and both arcs have the same radius, the
	portion of $\gamma_j$ between $a_j$ and $u$ is below the portion of $\gamma_{j+1}$
	between $a_{j+1}$ and $u$. Since $q$ is not above the
	portion of $\gamma_j$ between $a_j$ and $u$, $q$ cannot be above the portion of $\gamma_{j+1}$
	between $a_{j+1}$ and $u$. As $x(q)<x(u)$, this implies that $q$ cannot
	be above $\gamma_{j+1}$ and thus $p_{j+1}$ cannot be in $D_q$.

	\item
	If $\gamma'_j$ and $\gamma'_{j + 1}$ do not intersect, then both the right endpoint
	$b_j$ of $\gamma_j$ and the left endpoint $a_{j+1}$ of $\gamma_{j+1}$ are
	vertices of $\calU_1$ and $x(b_j)<x(a_{j+1})$. As $q$ is spanned by
	$\gamma_i'$ and $i<j$, $x(q)\leq x(b_j)$, and thus $x(q)<x(a_{j+1})$.
	Hence, $q$ cannot be above $\gamma_{j+1}$ and therefore $p_{j+1}$ cannot be in
	$D_q$.
	\end{itemize}

	The above proves that $p_{j+1}$ cannot be in $D_q$. Following the same
	analysis, we can show that $p_h$ cannot be in $D_q$ for all
	$h=j+2,j+3,\ldots,t$.

    Clearly, the algorithm runs in $O(k)$ time, where $k=|Q_1\cap D_q|$. This proves the lemma.
\end{proof}

By Lemma~\ref{lem:SearchInOneLayer}, if we store arcs of $\calU_1$ by a balanced
binary search tree, given a query point $q$ below $\ell$, the arc of $\calU_1$
spanning $q$ can be computed in $O(\log n)$ time and consequently $Q_1\cap D_q$ can be
reported in additional $O(|Q_1\cap D_q|)$ time. Recall that our goal is to
report $Q\cap D_q$. To report the remaining points, i.e., those of $Q\setminus
Q_1$ in $D_q$, we apply the idea recursively on $Q\setminus Q_1$. Specifically,
define $\calU_2$ as the lower envelope of the arcs of $\calA$ after the arcs
defined by the points of $Q_1$ are removed; let $Q_2$ denote the set of centers
of the arcs of $\calU_2$. In general, define $\calU_i$ as the lower envelope of
the arcs of $\calA$ after the arcs defined by the points of $\bigcup_{j=1}^{i-1}
Q_j$ are removed for $i=2,3,\ldots$ (see Fig.~\ref{fig:MultipleOneLowerEnvelopes}); let $Q_i$ denote the set of centers of the arcs of $\calU_i$.
We call $\{\calU_i\}$ the {\em lower envelope layers} of $\calA$. The
following theorem, which will be proved in Section~\ref{sec:ComputingLayers},
computes the lower envelope layers.

\begin{figure}[t]
    \centering
    \includegraphics[height=1.0in]{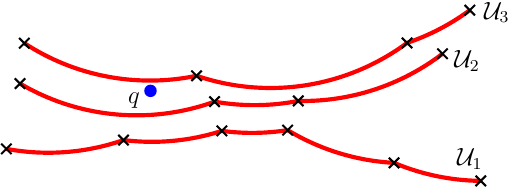}
    \caption{Illustrating layers of lower envelopes $\calU_1, \calU_2, \calU_3$.}
    \label{fig:MultipleOneLowerEnvelopes}
\end{figure}


\begin{theorem}
\label{theorem:ComputingLayers}
    The lower envelope layers of $\calA$ can be computed in $O(n\log n)$ time and $O(n)$ space, where $n=|\calA|$.
\end{theorem}

\paragraph{Proving Lemma~\ref{lem:lineUDRR}.}
We now have all ingredients to prove Lemma~\ref{lem:lineUDRR}. We compute the lower
envelope layers of $\calA$ by Theorem~\ref{theorem:ComputingLayers}. Then, we
construct a fractional cascading data structure on the vertices of the lower
envelope layers~\cite{ref:ChazelleFr86,ref:ChazelleFr862}.
This finishes the preprocessing, which takes $O(n)$ space and
$O(n\log n)$ time in total. Given a query unit disk $D_q$ centered at a point $q$
below the line $\ell$, using the fractional cascading data structure, we can
compute the arc of $\calU_1$ spanning $q$ in $O(\log n)$ time and compute the
arc of the next layer $\calU_2,\calU_3,\ldots$ spanning $q$ in $O(1)$ time each.
We compute the arc $\gamma'_i$ of $\calU_i$ that spans $q$ for all $i=1,2,\ldots$
until an index $j$ such that $q$ is below $\gamma_j'$
(and thus $Q_j$ does not have any point in
$D_q$, which is also the case for $Q_{j+1}, Q_{j+2}, \cdots$). Then, for each $\calU_i$ with $1\leq i\leq j-1$, using the arc
$\gamma_i'$, we apply Lemma~\ref{lem:SearchInOneLayer} to report the points of
$Q_i\cap D_q$. Because $q$ is above $\calU_i$ for each $1\leq i\leq j-1$, $Q_i$ has
at least one point in $D_q$. As such, the total time of the query algorithm is
bounded by $O(k+\log n)$, where $k=|Q\cap D_q|$.
This proves Lemma~\ref{lem:lineUDRR}.

\section{Computing layers of lower envelopes}
\label{sec:ComputingLayers}

In this section, we prove Theorem \ref{theorem:ComputingLayers}. We follow the
same notation as before, e.g., $Q$, $\calA$, $\calU_i$, $Q_i$, except that we now use $n$ to
denote $|Q|$ for convenience. Recall that all points of $Q$ are contained in a unit disk
and thus the distance of every two points of $Q$ is at most $2$.
For ease of exposition, we assume that no two points of $Q$ have the same $x$-coordinate.
For any subset $Q'\subseteq Q$, define $\calA(Q')=\{\gamma_p\ |\ p\in Q'\}$.

\begin{figure}[t]
    \centering
    \begin{minipage}[t]{0.41\textwidth}
    \centering
    \includegraphics[height=1.0in]{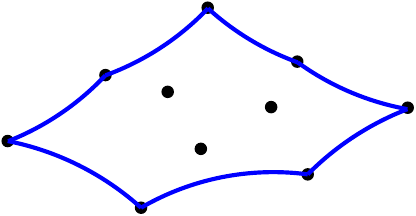}
    \caption{Illustrating the $\alpha$-hull of $Q$, for $\alpha = -1$.}
    \label{fig:AlphaHull}
    \end{minipage}
    \hspace{0.05in}
    \begin{minipage}[t]{0.55\textwidth}
    \centering
      \includegraphics[height=1.6in]{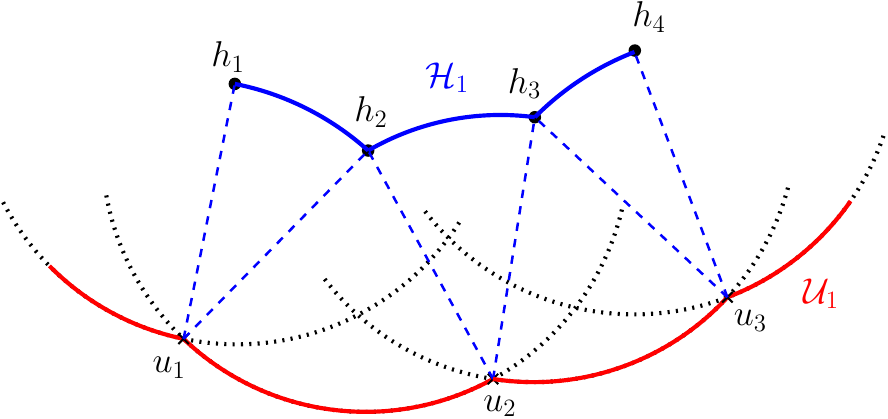}
    \caption{Illustrating the lower $\alpha$-hull $\calH_1$ of $Q$ and the lower
	envelope $\calU_1$ of $\calA$. Black dotted
	arcs are boundaries of underlying disks of arcs of $\calU_1$.
	Vertices of $\calH_1$ are centers of arcs of $\calU_1$, and vice versa.
	}
    \label{fig:LowerAlphaHullLayersToLowerEnvelopes}
    \end{minipage}
\end{figure}

Our goal is to compute the lower envelope layers $\{\calU_i\}$. Instead of computing them directly, we consider a {\em dual problem}.
We borrow a concept {\em $\alpha$-hull} from \cite{ref:EdelsbrunnerOn83}, which is a generalization of the  convex hull. For a real number $\alpha$, a \emph{generalized disk} of radius $1 / \alpha$ is defined to be a disk of radius $1 / \alpha$ if $\alpha > 0$, the complement of a disk of radius $- 1 / \alpha$ if $\alpha < 0$, and a halfplane if $\alpha = 0$. The \emph{$\alpha$-hull} of $Q$ is the intersection of all generalized disks with radius $1 / \alpha$ that contain all points of $Q$ (see Fig.~\ref{fig:AlphaHull}).
For our problem, we are interested in the case $\alpha = -1$. Henceforth,
unless otherwise stated, $\alpha=-1$.

It is known that the leftmost (resp., rightmost) point of $Q$ must be the
leftmost (resp., rightmost) vertex of the $\alpha$-hull of
$Q$~\cite{ref:EdelsbrunnerOn83}. The \emph{lower $\alpha$-hull} of $Q$, denoted
by $\calH_1$, is defined as the portion of the boundary of the $\alpha$-hull
counterclockwise from its leftmost vertex to its rightmost vertex
(similar concepts have been used elsewhere, e.g.,~\cite{ref:DumitrescuSp22}).

For any two points $p$ and $p'$ of $Q$, as their distance is at most $2$, there
are two circular arcs of radius $1$ connecting them. One of these arcs having its
center below the line through $p$ and $p'$ while the other having its center
above the line (recall that $x(p)\neq x(p')$ due to our assumption);
we call the former arc the \emph{concave arc} of $p$ and $p'$,
denoted by $\gamma(p, p')$. Note that the lower $\alpha$-hull $\calH_1$ comprises
concave arcs~\cite{ref:EdelsbrunnerOn83}.


We observe the following {\em duality} between the lower hull $\calH_1$ of $Q$ and the lower envelope $\calU_1$ of $\calA$ (see Fig.~\ref{fig:LowerAlphaHullLayersToLowerEnvelopes}):
The center of each arc in $\calH_1$ is a vertex of $\calU_1$ while the center of each arc of $\calU_1$ is a vertex of $\calH_1$. Due to this duality, $Q_1$ is exactly the set of vertices of $\calH_1$.


Like the lower envelope layers of $\calA$, we can correspondingly define {\em lower
$\alpha$-hull layers} of $Q$. Specifically, define $\calH_2$ as the lower
$\alpha$-hull of $Q \setminus Q_1$, i.e., the remaining points of $Q$ after
vertices of $\calH_1$ are removed; $\calH_i$ is defined similarly for
$i=3,4,\ldots$; see
Fig.~\ref{fig:LowerAlphaHullLayers}. As above, each $\calH_i$ is dual to
$\calU_i$, and thus $Q_i$ is the set of vertices of $\calH_i$. As such, to
compute layers of lower envelopes $\{\calU_i\}$ of $\calA$, it suffices to compute layers of
lower $\alpha$-hulls $\{\calH_i\}$ of $Q$, which is our focus below.

\begin{figure}[t]
    \centering
    \includegraphics[height=0.9in]{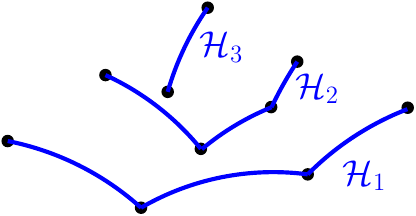}
    \caption{Illustrating lower $\alpha$-hull layers $\{\calH_1,
	\calH_2,\calH_3\}$.}
    \label{fig:LowerAlphaHullLayers}
\end{figure}

We present an algorithm to compute the lower $\alpha$-hull layers $\{\calH_i\}$ in $O(n)$ space and $O(n\log n)$ time.
We follow the scheme of Chazelle's algorithm~\cite{ref:ChazelleOn85} for computing convex hull layers of a set of points in the plane. Our algorithm is actually simpler since cross deletions are not needed in our algorithm.
The main idea is to construct a {\em tree graph} $G$ embedded in the plane such that each edge is a circular arc.
$\calH_1$ can be produced in $O(|\calH_1|)$ time by using $G$. Then, vertices of $\calH_1$
are removed from $G$ and $G$ is updated so that $\calH_2$ can be produced in
$O(|\calH_2|)$ time. Repeating this process until $G$ becomes $\emptyset$ will produce the lower
$\alpha$-hull layers $\{\calH_i\}$. In what follows, we first define the graph
$G$ in Section~\ref{sec:defgraph} and then describe an algorithm to construct it in
Section~\ref{sec:constructgraph}. Finally in Section~\ref{sec:computelayer} we compute lower $\alpha$-hull layers using $G$.

\subsection{Defining the tree graph $\boldsymbol{G}$}
\label{sec:defgraph}


Let $p_1, p_2, ..., p_n$ be the list of the points of $Q$ sorted from left to
right. Let $T$ be a complete binary tree whose leaves store $p_1, p_2, ...,
p_n$ from left to right, respectively. For each node $v$ of $T$, let $Q(v) \subseteq Q$ be the
set of points that are stored at the leaves of the subtree rooted at $v$ and let
$\calA(v)=\calA(Q(v))$. Let
$\calH(v)$ denote the lower $\alpha$-hull of points in $Q(v)$ and $\calU(v)$ the
lower envelope of $\calA(v)$. Hence, $\calH(v)$ and $\calU(v)$ are dual to each
other.

The graph $G$ is defined as follows: Its vertex set is $Q$ and its edge set
consists of arcs of $\calH(v)$ of all nodes $v$ of $T$ (see
Fig.~\ref{fig:GraphG}). As such, each edge of $G$ is a
concave arc.

For any vertex $p$ of the lower $\alpha$-hull $\calH$ of a subset $Q'$ of $Q$, we say that a
circular arc $\gamma$ containing $p$ is {\em tangent} to $\calH$ at $p$ if no point of $Q'$
is contained in the interior of the underlying disk of $\gamma$. Note that $\gamma$ is
tangent to $\calH$ if and only if the two adjacent vertices of $p$ on $\calH$
are outside the underlying disk of $\gamma$.

\begin{figure}[t]
    \centering
    \begin{minipage}[t]{0.48\textwidth}
        \centering
        \includegraphics[height=1.0in]{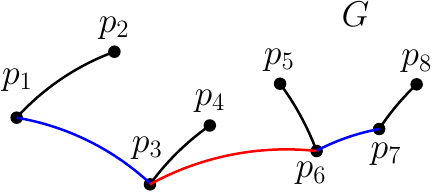}
        \caption{Illustrating the graph $G$ for a set $Q = \{p_1, p_2, ..., p_8\}$ of $8$ points.}
        \label{fig:GraphG}
    \end{minipage}
    \hspace{0.08in}
    \begin{minipage}[t]{0.48\textwidth}
        \centering
        \includegraphics[height=1.0in]{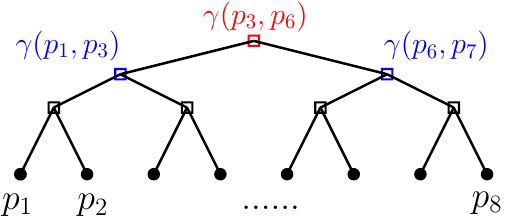}
        \caption{Illustrating $T$ for the example in Fig.~\ref{fig:GraphG}. Internal
		nodes store common tangent arcs, which are edges of $G$.}
        \label{fig:TreeT}
    \end{minipage}
\end{figure}




Consider a node $v\in T$.
Let $u$ and $w$ be $v$'s left and right children, respectively. A concave arc
$\gamma(p_i,p_j)$
connecting a vertex $p_i$ of $\calH(u)$ and a vertex $p_j$ of $\calH(w)$ is
called a {\em common tangent arc} of $\calH(u)$ and $\calH(w)$ if $\gamma(p_i,p_j)$ is tangent
to $\calH(u)$ at $p_i$ and tangent to $\calH(w)$ at $p_j$. By duality,
$\gamma(p_i,p_j)$ corresponds to the intersection $a$ between $\calU(u)$ and $\calU(w)$ (i.e., $a$
is the center of $\gamma(p_i,p_j)$). It has been proved
in Lemma 10 of~\cite{ref:WangCo22} that $\calU(u)$ and $\calU(w)$ have at most one
intersection, and thus $\calH(u)$ and $\calH(w)$ have at most one common
tangent arc. In fact, since all points of $Q$ are contained in a unit disk,
$\calH(u)$ and $\calH(w)$ have exactly one common tangent arc, say, $\gamma(p_i,p_j)$, connecting a
vertex $p_i$ of $\calH(u)$ and a vertex $p_j$ of $\calH(w)$. Then $\calH(v)$
consists of the following three portions in order from left to right: the
portion of $\calH(u)$ between its leftmost vertex and $p_i$, the arc $\gamma(p_i,p_j)$, and
the portion of $\calH(w)$ between $p_j$ and its rightmost vertex. We store
$\gamma(p_i,p_j)$ at $v$, denoted by $\gamma(v)$; see Fig.~\ref{fig:TreeT}.
The common tangent arcs $\gamma(v)$ for all internal nodes $v$ of $T$ form exactly the edge set of $G$.

We store the graph $G$ in an adjacency-list structure as follows. Each vertex $p$ of $G$
is associated with two doubly linked lists
$L_l(p)$ and $L_r(p)$ such that $L_l(p) \cup L_r(p)$ contains all adjacent
vertices of $p$ in $G$, where $L_l(p)$ (resp., $L_r(p)$) stores adjacent vertices
of $p$ that are to the left (resp., right) of $p$. For each adjacent vertex $q$
of $p$, we define the \emph{tangent angle} of the concave arc $\gamma(p,q)$ of $G$
connecting $p$ and $q$ as the acute angle of the tangent ray of $\gamma(p, q)$ at $p$
following the direction toward $q$ with the horizontal line through $p$ (see
Fig.~\ref{fig:AdjacencyList}). Vertices of $L_l(p)$ (resp., $L_r(p)$) are sorted
by the tangent angles of their corresponding arcs. The \emph{bottom edge} of $L_l(p)$
(resp., $L_r(p)$) is defined as the arc with the minimum tangent angle in $L_l(p)$ (resp., $L_r(p)$); see Fig.~\ref{fig:AdjacencyList}.
We add two pointers at $p$ to access the two bottom edges in $L_l(p)$
and $L_r(p)$.

\begin{figure}[t]
    \centering
    \includegraphics[height=1.3in]{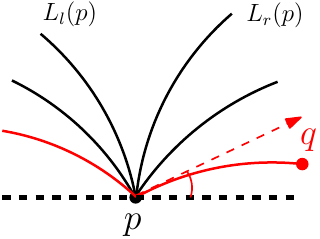}
    \caption{Illustrating the adjacency lists $L_l(p)$ and $L_r(p)$ at $p$.
	The two red arcs are bottom edges. The red dashed segment with arrow is the tangent ray of $\gamma(p,q)$ at $p$ and the tangent angle is shown.}
    \label{fig:AdjacencyList}
\end{figure}

\subsection{Constructing the tree graph $\boldsymbol{G}$}
\label{sec:constructgraph}

The following lemma will be used as a subroutine in our algorithm for constructing $G$.

\begin{lemma}
    \label{lem:CommonArcTangent}
    Given the lower $\alpha$-hull $\calH'$ of a subset $Q'\subseteq Q$ and the
	lower $\alpha$-hull $\calH''$ of another subset $Q''\subseteq Q$ such that $Q'$
	and $Q''$ are separated by a vertical line, the common
	tangent arc of $\calH'$ and $\calH''$ can be computed in $O(|\calH'| + |\calH''|)$ time.
\end{lemma}
\begin{proof}
     Without loss of generality, we assume that $\calH'$ is to the left of
	 $\calH''$. Our goal is to compute a vertex $u\in \calH'$ and a vertex $v\in
	 \calH''$ such that the arc $\gamma(u,v)$ is tangent to both $\calH'$ and
	 $\calH''$. The algorithm is similar to that for computing a common tangent
	 of two lower convex hulls that are separated by a vertical line; we briefly
	 discuss it below.

    Initially we set $u$ to the rightmost vertex of $\calH'$ and $v$ the
	leftmost vertex of $\calH''$. We keep moving $u$ leftwards on $\calH'$ until
	$\gamma(u,v)$ is tangent to $\calH'$ at $u$. Then we check whether $\gamma(u,v)$ is
	tangent to $\calH''$ at $v$. If yes, then we are done.
	Otherwise, we keep moving $v$ rightwards on $\calH''$ until $\gamma(u,v)$ is tangent to
	$\calH''$ at $v$. Next we check whether $\gamma(u,v)$ is tangent to $\calH'$ at
	$u$. If yes, we are done. Otherwise, we move $u$
	leftwards again. We repeat this process and eventually a common tangent arc
	will be found. Clearly, the runtime is $O(|\calH'| + |\calH''|)$.
\end{proof}

With Lemma~\ref{lem:CommonArcTangent}, the next lemma constructs the graph $G$.

\begin{lemma}
    \label{lem:ConstructG}
    The graph $G$ can be constructed in $O(n \log n)$ time and $O(n)$ space.
\end{lemma}
\begin{proof}
    As the vertex set of $G$ is $Q$, our goal is to construct all edges and
	store them in the adjacent-list structure, i.e., for each point $p\in Q$,
	construct the lists $L_l(p)$ and $L_r(p)$. To this end,
	our algorithm proceeds following the tree $T$ in a bottom-up manner.

For each vertex $v\in T$, we define $G(v)$ as the graph $G$ but only on the
points of $Q(v)$.  As such, $G(v)$ is $G$ if $v$ is the root and
$G(v)=\emptyset$ if $v$ is a leaf.

Consider an internal node $v$ of $T$, with $u$ and $w$ as its left and right children,
respectively.
We assume that $G(u)$ and $G(w)$ have been computed, i.e., for each point $p$ of
$Q(u)$ (resp., $Q(v)$), we have two corresponding lists $L_l(p)$ and $L_r(p)$ with respect to $G(u)$
(resp., $G(v)$). Next, we construct $G(v)$ using $G(u)$ and $G(w)$.

Observe that $G(v)$ is the union of $G(u)$, $G(w)$, and the
common tangent arc of $\calH(u)$ and $\calH(w)$, denoted by $\gamma(p,q)$, with $p\in \calH(u)$ and $q\in
\calH(w)$.  Since $Q(u)$ and $Q(w)$ are separated by a vertical line,
we can compute the arc $\gamma(p,q)$ in $O(|Q(u)|+|Q(w)|)$ time by Lemma~\ref{lem:CommonArcTangent}.
Note that we can traverse on $\calH(u)$ (resp., $\calH(w)$) in constant time per vertex using the
bottom edge pointers of vertices in $G(u)$ (resp., $G(w)$).
Observe that the arc $\gamma(p,q)$ must be the bottom edge in $L_r(p)$ as well as $L_l(q)$
in $G(v)$. As such, we simply add $q$ to the bottom of the current list $L_r(p)$
and add $p$ to the bottom of the current list $L_l(q)$, and also update the
bottom edge pointers of $p$ and $q$ accordingly.
In this way, $G(v)$ can be computed in $O(|Q(u)| + |Q(w)|)$ time, or in
$O(|Q(v)|)$ time as $Q(v)=Q(u)\cup Q(w)$. Hence, the total time for constructing the
graph $G$ is $O(n \log n)$ and the space complexity is $O(n)$.
\end{proof}

\subsection{Computing lower $\boldsymbol{\alpha}$-hull layers}
\label{sec:computelayer}

We next use the graph $G$ to compute the lower $\alpha$-hull layers $\{\calH_i\}$. 

First of all, $\calH_1$ can be obtained in $O(|\calH_1|)$ time by using bottom edge pointers of $G$, say, starting from the leftmost point of $Q$, which is the leftmost vertex of $\calH_1$, since arcs of $\calH_1$ must be bottom edges of vertices of $\calH_1$. Then, we remove vertices of $\calH_1$ (along with their incident edges) from $G$. Using the updated $G$, the second layer lower $\alpha$-hull $\calH_2$ can be computed in $O(|\calH_2|)$ time similarly. We repeat this process until $G$ becomes empty. The following lemma shows that removing a vertex from $G$ can be done in $O(\log n)$ amortized time.

\begin{lemma}
    \label{lem:ComputingLowerAlphaHullLayers}
    All point deletions in the entire algorithm can be done in $O(n\log n)$ time and $O(n)$ space.
\end{lemma}
\begin{proof}
\label{proof:lem-ComputingLowerAlphaHullLayers}
Suppose we want to delete a point $p$ from $G$ and $p$ is a vertex of the lower $\alpha$-hull of $G$.
The deletion of $p$ will result in the removal of all arcs of $G$ connecting $p$.
In addition, new arcs may be computed as well.

Let $\{v_1, v_2, ..., v_g\}$ be the list of the nodes of $T$ encountered when traversing from the leaf node storing the point $p$ to the root of $T$. The deletion of $p$ may affect lower $\alpha$-hulls $\calH(v_i)$, for $i = 1, 2, ..., g$. We will update $G(v_i)$ (and thus $\calH(v_i)$) for $i = 1, 2, ..., g$ in this order.

Consider a node $v_i$ with $2\leq i\leq g$. Note that $v_{i-1}$ is a child of
$v_i$. Let $v$ refer to the child of $v_i$ other than $v_{i - 1}$. Depending on
whether $p$ is an endpoint of the arc $\gamma(v_i)$ stored at $v_i$, i.e., the common tangent
arc of $\calH(v_{i - 1})$ and $\calH(v)$, there are two cases.
If $p$ is not an endpoint of $\gamma(v_i)$, then removing $p$ does not affect
$\gamma(v_i)$ as well as $\gamma(v_j)$ for any $i+1\leq j \leq g$.
Hence, in this case, we are done with deleting $p$.
In the following, we focus on the case
where $p$ is an endpoint of $\gamma(v_i)$ (see
Fig.~\ref{fig:NewCommonArcTangent}). Below we only discuss the case where $p$ is the left
endpoint of $\gamma(v_i)$ since the other case is symmetric. Let $c$ be the other
endpoint of $\gamma(v_i)$ and to be more informative we use $\gamma(p,c)$ to refer to $\gamma(v_i)$.

    \begin{figure}[t]
    \centering
    \includegraphics[height=1.2in]{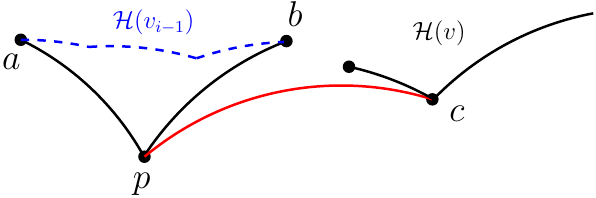}
    \caption{$p$ is an endpoint of $\gamma(v_i)$, i.e., the common tangent arc (the red arc) of the new $\calH(v_{i - 1})$ and $\calH(v)$.}
    \label{fig:NewCommonArcTangent}
    \end{figure}

Note that each arc of $L_l(p)\cup L_r(p)$ is $\gamma(v_j)$ for some $j\in [1,g]$.
Let $\gamma(a, p)$ be the last arc that has been processed due to the deletion of $p$ with $p$ as the right endpoint of the arc, and $\gamma(p,b)$ the last arc that has been processed with $p$ as the left endpoint of the arc (see Fig.~\ref{fig:NewCommonArcTangent}).
We assume that both $a$ and $b$ are well-defined (otherwise the algorithm is similar but simpler). Note that $\gamma(a, p)$ and $\gamma(p,b)$ are actually arcs of the old $\calH(v_{i-1})$ before $v_{i-1}$ is processed.
Since we process nodes of $T$ in a bottom-up matter, $a$ and $b$ can be accessed
from $L_l(p)$ and $L_r(p)$ in constant time.
Observe that the portion of the new lower
$\alpha$-hull $\calH(v_{i - 1})$ between $a$ and $b$ must lie above the ``wedge''
formed by $\gamma(a, p)$ and $\gamma(p, b)$ (see Fig.~\ref{fig:NewCommonArcTangent}).
Our goal is to compute a new common tangent arc $\gamma(s, t)$ of the new $\calH(v_{i-1})$ and $\calH(v)$, with
$s \in \calH(v_{i-1})$ and $t \in \calH(v)$, as follows.

    \begin{figure}[t]
    \centering
    \begin{minipage}[t]{0.48\textwidth}
        \centering
        \includegraphics[height=1.3in]{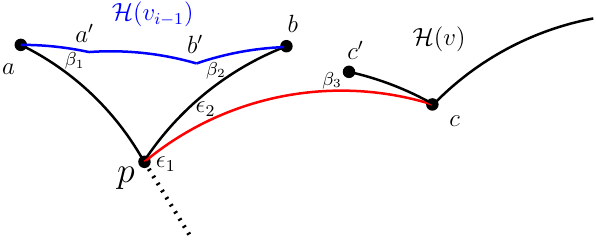}
        \caption{Illustrating points $a'$, $b'$ and $c'$, angles $\{\beta_1, \beta_2, \beta_3\}$ and $\{\epsilon_1, \epsilon_2\}$.}
        \label{fig:FiveAngles}
    \end{minipage}
    \hspace{0.08in}
    \begin{minipage}[t]{0.48\textwidth}
        \centering
        \includegraphics[height=1.4in]{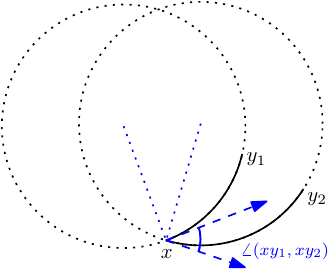}
        \caption{Illustrating angle $\angle(x y_1, x y_2)$ of arcs $\gamma(x, y_1)$ and $\gamma(x, y_2)$. Blue rays with arrows are tangent rays of $\gamma(x,y_1)$ and $\gamma(x,y_2)$ at $x$.}
        \label{fig:AngleOfArcs}
    \end{minipage}
\end{figure}

    Observe that $s$ must lie between $a$ and $b$ on $\calH(v_{i - 1})$ and $t$ is to the left of $c$ on $\calH(v)$. We define $a'$ as the right adjacent vertex of $a$ and $b'$ as the left adjacent vertex of $b$
    on $\calH(v_{i-1})$  (see Fig.~\ref{fig:FiveAngles}). Let $c'$ be the left adjacent vertex of $c$ on $\calH(v)$. The degenerate case in which $a' = b'$, or $a' = b$ and $b' = a$ can be handled trivially. Since $p$ is a vertex of the current lower $\alpha$-hull of $G$, arcs $\gamma(a, p)$, $\gamma(b, p)$, and $\gamma(c, p)$ are bottom edges of $L_r(a)$, $L_l(b)$, and $L_l(c)$, respectively. We can also access $a'$, $b'$, and $c'$ 
    in constant time.

To describe our algorithm for computing $\gamma(s,t)$, we define {\em the angle} $\angle(x y_1, x y_2)$ of two arcs $\gamma(x, y_1)$ and $\gamma(x, y_2)$ as follows.
For each $j=1,2$, define $\rho_j$ as the ray from $x$ toward $y_j$ and tangent to the underlying disk of $\gamma(x, y_j)$ at $x$.
$\angle(x y_1, x y_2)$ is defined as the angle between $\rho_1$ and $\rho_2$ (see Fig.~\ref{fig:AngleOfArcs}). Our algorithm considers the following five angles, $\beta_1 = \angle(a a', a p)$, $\beta_2 = \angle(b b', b p)$, $\beta_3 = \angle(c c', c p)$, $\epsilon_1 = \angle(p a'', p c)$, and $\epsilon_2 = \angle(p b, p c)$, where $a''$ is a point on the extension of arc $\gamma(a, p)$ (see Fig.~\ref{fig:FiveAngles}).

    \begin{figure}[t]
    \centering
    \begin{minipage}[t]{0.48\textwidth}
        \centering
        \includegraphics[height=1.3in]{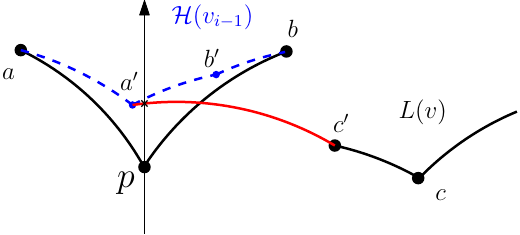}
        \caption{Illustrating the case of pulling up $p$ in which $\epsilon_1$ becomes null.}
        \label{fig:Epsilon1Null}
    \end{minipage}
    \hspace{0.08in}
    \begin{minipage}[t]{0.48\textwidth}
        \centering
        \includegraphics[height=1.3in]{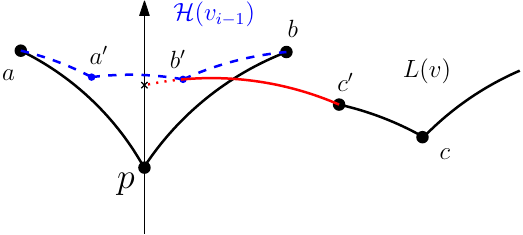}
        \caption{Illustrating the case of pulling up $p$ in which $\epsilon_2$ becomes null.}
        \label{fig:Epsilon2Null}
    \end{minipage}
\end{figure}

    Our algorithm for computing $\gamma(s,t)$ can be viewed as a process of ``pulling
	up'' $p$ vertically until $p$ disappears in the new lower $\alpha$-hull
	$\calH(v_i)$. This happens when one of $\{\epsilon_1, \epsilon_2\}$ becomes
	null (see Fig.~\ref{fig:Epsilon1Null} and~\ref{fig:Epsilon2Null}). If one of the angles of $\{\beta_1, \beta_2, \beta_3\}$ becomes
	null, then we will update $x$ and $x'$, $x \in \{a, b, c\}$ accordingly to
	obtain new $\beta$-angles. More specifically,  if $\angle(xx', xp)$, $x \in
	\{a, b, c\}$ becomes null, then we reset $x$ to $x'$, and reset $x'$ to the
	left (if $x \in \{b, c\}$) or right (if $x \in \{a\}$) neighbor of the old
	$x'$. For the purpose of time analysis, we say that the old $x$ is \emph{wrapped}.
    We can avoid calculating those five angles by computing the intersections $a^*$, $b^*$, and $c^*$ of the vertical line through $p$ with extensions of arcs $\gamma(a, a')$, $\gamma(b, b')$ and $\gamma(c, c')$, respectively (see Fig.~\ref{fig:Deletion}). The lowest point of $a^*$, $b^*$, and $c^*$ is the next candidate location of $p$.
    Before moving $p$ to the next location, we check whether $\{a, p, c\}$ or $\{p, b, c\}$ will be on the same unit circle during the movement of $p$, and a positive answer implies that either $\epsilon_1$ or $\epsilon_2$ is null.
    We iterate this process until one of $\epsilon_1$ and $\epsilon_2$ is null.
	Once the common tangent arc $\gamma(s,t)$ is computed, we proceed on processing
	$v_{i+1}$.
    Note that $p$ is in $L_l(c)$ and is actually the bottom edge since $p$ is a vertex of the lower $\alpha$-hull of $G$; as such, we can remove $p$ from $L_l(c)$ and reset its bottom edge in constant time.

    \begin{figure}[t]
        \centering
        \includegraphics[width=3.7in]{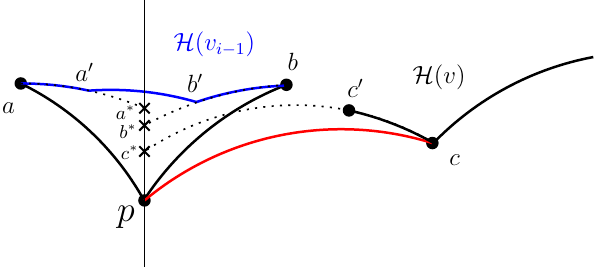}
        \caption{Illustrating the definitions of $a^*$, $b^*$, and $c^*$.}
        \label{fig:Deletion}
    \end{figure}


   The running time of the algorithm is linear in the number of wrapped
   vertices on $\calH(v_{i-1})$ and $\calH(v)$.
   If a vertex $u$ is wrapped by $a$ or $c$, then $u$ becomes a vertex on the
   new $\calH(v_i)$. We call this
   wrapping step a \emph{promotion} (because $u$ used to be a vertex of
   $\calH(v_{i-1})$ and not a vertex of $\calH(v_i)$,
   but now is ``promoted'' to be a vertex of $\calH(v_i)$).
   Since the height of $T$ is $O(\log n)$, the total number of
   promotions for deleting all points $p\in G$ is
   bounded by $O(n \log n)$. On the other hand, if a vertex $u$ is
   wrapped by $b$, we call it a \emph{confirmation}. A critical observation is
   that the previous wrapping on $u$ during the deletion of $p$ must be a
   promotion (i.e., during processing $v_{i-1}$, $u$ was wrapped as a promotion). Consequently, any confirmation must be immediately preceded by a
   promotion.
   As such, the total number of confirmations for deleting all points $p\in G$
   is no more than that of promotions, which is $O(n\log n)$.
   Therefore, the overall time of the algorithm for deleting all points $p\in G$
   is bounded by $O(n \log n)$. The space complexity of the algorithm is $O(n)$.
\end{proof}

With Lemma~\ref{lem:ComputingLowerAlphaHullLayers}, Theorem~\ref{theorem:ComputingLayers} is proved.

\section{Concluding remarks and the dynamic unit-disk range emptiness queries}
\label{sec:ConcludingRemarks}

In this paper, we presented a dynamic data structure for unit-disk range reporting queries. Our query algorithm achieves optimal complexity, improving the previous result. While we mostly follow the previous algorithmic scheme, one main ingredient is a shallow cutting algorithm for circular arcs that may be interesting in its own right. We also proposed a static data structure whose complexities are optimal and match those of the previously best result; our approach is much simpler than the previous work. Our techniques may be extended to solve other related problems about unit disks. In the following, We demonstrate one exemplary problem: the unit-disk range emptiness queries and its dynamic version.

\paragraph{The static problem.}
Let $P$ be a set of $n$ points in the plane. The problem is to build a data structure to answer the following {\em unit-disk range emptiness queries}: Given a unit disk $D$, determine whether $D$ contains a point of $P$, and if so, return such a point. By computing the Voronoi diagram of $P$ and then constructing a point location data structure on the diagram~\cite{ref:EdelsbrunnerOp86,ref:KirkpatrickOp83,ref:SarnakPl86}, one can build a data structure of $O(n)$ space in $O(n\log n)$ time with $O(\log n)$ query time. Using our techniques in this paper, we can provide alternative (and slightly simpler) solution with asymptotically the same complexities as follows.

We follow our method for the original unit-disk range reporting queries but only maintain the lower envelope of $\calA$ (i.e., no need to compute all layers of lower envelopes and thus the algorithm becomes much simpler). The preprocessing still takes $O(n \log n)$ time and $O(n)$ space. 

Given a query unit disk $D_q$ centered at a point $q$, we first check whether $q$ is in a cell of $\calC$ by Lemma~\ref{lem:grid}(2). If no cell of $\calC$ contains $q$, then we return null. Otherwise, let $C$ be the cell of $\calC$ that contains $q$. If $P(C)\neq \emptyset$, then all points of $P(C)$ are in $D_q$ and thus we return an arbitrary point of $P(C)$. If $P(C)= \emptyset$, then for each $C'\in N(C)$, we solve a line-separable problem for $P(C')$.
More specifically, suppose that $C'$ and $C$ are separated by a horizontal line $\ell$ with $C'$ above $\ell$. Then, 
a point of $P(C')$ is in $D_q$ if and only if the center $q$ is above the lower envelope of $\calA$ (defined by $P(C')$ with respect to $\ell$), which can be determined by performing binary search with $q$ on the lower envelope. The total query time is thus $O(\log n)$. 

\paragraph{The dynamic problem.}
In the dynamic problem, point insertions and deletions are allowed for $P$. One can solve the problem by using a dynamic nearest neighbor search data structure (i.e., given a query disk $D$, using a nearest neighbor query we find a point $p\in P$ nearest to the center of $D$; $D$ contains a point of $P$ if and only if $p\in D$). The current best dynamic nearest neighbor search data structure is given by Chan~\cite{ref:ChanDy20}; with that, we can obtain a data structure of $O(n)$ space in $O(n\log n)$ time that supports $O(\log^2 n)$ amortized insertion time, $O(\log^4n)$ amortized deletion time, and $O(\log^2 n)$ time for unit-disk range emptiness queries. In the following, using our techniques, we propose a better result. 

As in Section~\ref{sec:dynamicreport} for the dynamic reporting problem, we can use Lemma~\ref{lem:dynamiccell} to reduce the problem to the following line-separated problem. 

\begin{problem}(Dynamic line-separable unit-disk range emptiness queries)
\label{problem:dynamic-LS-rangeempty}
Given a set $Q$ of $m$ points above a horizontal line $\ell$, 
build a data structure to maintain $Q$ to support the following operations. 
(1) Insertion: insert a point to $Q$; (2) deletion: delete a point from $Q$; (3) unit-disk range emptiness query: given a unit disk $D$ whose center is below $\ell$, determine whether $D$ contains a point of Q, and if so, return such a point. 
\end{problem}

To solve the line-separable problem, we define the set $\calA$ of arcs using $Q$ in the same way as before. Let $D_q$ be a unit disk with center $q$ below $\ell$. Note that $D_q\cap Q\neq \emptyset$ if and only if $q$ is above the lower envelope of $\calA$. Further, $q$ is above the lower envelope of $\calA$ if and only if the lowest arc of $\calA$ intersecting $\ell_q$ is below $q$, where $\ell_q$ is the vertical line through $q$. Therefore, our problem reduces to the following {\em vertical line queries} subject to arcs insertions and deletions for $\calA$: Given a vertical line $\ell^*$, find the lowest arc of $\calA$ that intersects $\ell^*$. 

To solve the dynamic vertical line query problem among arcs of $\calA$, we apply Chan's framework~\cite{ref:ChanDy01} for the dynamic vertical line query problem among a set of lines (in the dual plane, a vertical line query is dual to the following problem: Finding an extreme point on the convex hull of all dual points along a query direction). To this end, we need the following two components: (1) a dynamic data structure of $O(m)$ space with $O(\log m)$ query time and $m^{O(1)}$ update time; (2) a deletion-only data structure of $O(m)$ space that can be built in $O(m\log m)$ time, supporting $O(\log m)$ query time and 
$O(\log m)$ amortized deletion time. For (1), we can use our static data structure as discussed above, i.e., whenever there is an update, we simply rebuild the data structure. For (2), Wang and Zhao~\cite{ref:WangCo22} already provided such a data structure. Using these two components, we can apply exactly the same framework of Chan~\cite{ref:ChanDy01}. Indeed, the framework still works for the arcs of $\calA$ because every arc is $x$-monotone. With the framework and the above two components, we can obtain a data structure of $O(m)$ space that allows insertions and deletions of arcs of $\calA$ in $O(\log^{1+\epsilon}m)$ amortized update time and answers a vertical line query in $O(\log m)$ time, where $m$ is the size of the current set $\calA$.
Consequently, we can solve Problem~\ref{problem:dynamic-LS-rangeempty} with the same time complexities. Finally, with Lemma~\ref{lem:dynamiccell} and our problem reduction, we can have a data structure of $O(n)$ space that allows insertions and deletions of points of $P$ in $O(\log^{1+\epsilon}n)$ amortized time and answers a unit-disk range emptiness query in $O(\log n)$ time, where $n$ is the size of the current set $P$.



\end{document}